\RequirePackage{fix-cm}
\documentclass[smallextended]{svjour3}       % onecolumn (second format)
\smartqed  % flush right qed marks, e.g. at end of proof
\usepackage{graphicx}

\usepackage{amsmath,amssymb,array}
\usepackage{siunitx}
\usepackage{longtable,tabularx}
\usepackage{algorithm}
\usepackage{algorithmic}
\usepackage{multirow}
\usepackage{color}

\newcommand{\argmax}{\mathop{\rm argmax}\limits}
\newcommand{\argmin}{\mathop{\rm argmin}\limits}

\def\cX{\mathcal X}

 \newenvironment{proofof}[1]{\vspace*{5mm} \par \noindent
{\it Proof of #1:\hspace{2mm}}}{\qed
}
\def\Label#1{\label{#1}\ [\ \text{#1}\ ]\ }
\def\Label{\label}
\begin{document}
\title{Iterative minimization algorithm on a mixture family}
\titlerunning{Iterative minimization algorithm on a mixture family}
\author{Masahito~Hayashi}

\institute{M. Hayashi \at
School of Data Science, The Chinese University of Hong Kong, Shenzhen, Longgang District, Shenzhen, 518172, China,
International Quantum Academy (SIQA), Futian District, Shenzhen 518048, China,
and
Graduate School of Mathematics, Nagoya University, Chikusa-ku, Nagoya 464-8602, Japan.
\\
              \email{e-mail: hmasahito@cuhk.edu.cn, hayashi@iqasz.cn}           %  \\
}

\date{Received: date / Accepted: date}

\maketitle

\begin{abstract}
Iterative minimization algorithms appear in various areas including machine learning, neural networks, and information theory.
The em algorithm is one of the famous iterative minimization algorithms in the area of machine learning, and the Arimoto-Blahut algorithm is a typical iterative algorithm in the area of information theory.
However, these two topics had been separately studied for a long time. In this paper, we generalize an algorithm that was recently proposed in the context of the Arimoto-Blahut algorithm.
Then, we show various convergence theorems, one of which covers the case when
each iterative step is done approximately.
Also, we apply this algorithm to the target problem of the em algorithm, and propose 
its improvement. In addition, we apply it to other various problems in information theory.
\end{abstract}


\section{Introduction}\Label{S1}
Optimization over distributions is an important topic in various areas.
For example, the minimum divergence between a mixture family and an exponential family 
has been studied by using the em algorithm
in the areas of machine learning and neural networks \cite{Amari1,Amari,Fujimoto,Allassonniere}.
The em algorithm is an iterative algorithm to calculate the above minimization
and it is rooted in the study of Boltzmann machines \cite{Bol}.
In particular, the paper \cite{Fujimoto} formulated 
the em algorithm under
the framework with Bregman divergence \cite{Amari-Nagaoka,Amari-Bregman}.
The topic of the em algorithm 
has been mainly studied in the community of machine learning, neural networks, and information geometry.
As another iterative algorithm, the Arimoto-Blahut algorithm is known as an algorithm to maximize the mutual information by changing the distribution on the input system \cite{Arimoto,Blahut}. 
This maximization is needed to calculate the channel capacity \cite{Shannon}.
This algorithm has been generalized to various settings 
including the rate distortion theory \cite{Blahut,Csiszar,Cheng,YSM}, the capacity of a wire-tap channel \cite{Yasui},
and their quantum extensions \cite{Nagaoka,Dupuis,Sutter,Li,RISB}.
In particular, the two papers \cite{YSM,RISB} made very useful generalizations to cover
various topics in information theory.
This topic has been mainly studied in the community of information theory.

However, only a limited number of studies have discussed the relation between the two topics, the em algorithm and the Arimoto-Blahut algorithm as follows.
The papers \cite{Tusnady,Sullivan} pointed out that 
the Arimoto-Blahut algorithm can be considered as an alternating algorithm
in a similar way to the EM and the em algorithms. 
Recently, the paper \cite{Shoji} pointed out that the maximization of the mutual information 
can be considered to be the maximization of 
the projected divergence to an exponential family by changing an element of 
a mixture family.
The paper \cite{reverse} generalized this maximization to the framework with 
Bregman divergence \cite{Amari-Nagaoka,Amari-Bregman} and applied this setting to 
various problems in information theory.
Also, the recent paper \cite{Bregman-em} applied the em algorithm to 
the rate-distortion theory, which is a key topic in information theory.

In this paper, we focus on a generalized problem setting proposed in \cite{RISB}, 
which is given as an optimization over the set of input quantum states.
As the difference from the former algorithm,
their algorithm \cite{RISB} has an acceleration parameter.
Changing this parameter, we can enhance the convergence speed under a certain condition.
To obtain wider applicability,
we extend their problem setting to the minimization over a general mixture family.
Although they discussed the convergence speed only when there is no local minimizer,
our analysis covers the convergence speed to a local minimizer
 even when there exist several local minimizers.
Further, since our setting covers a general mixture family as the set of input variables,
our method can be applied to the minimum divergence between a mixture family and an exponential family, which is the objective problem in the em algorithm.
That is, this paper presents a general algorithm including the em algorithm
as well as the Arimoto-Blahut algorithm.
This type of relation between the em algorithm and the Arimoto-Blahut algorithm
is different from the relation pointed by the papers \cite{Tusnady,Sullivan}.

There is a possibility that each iteration can be calculated only approximately.
To cover such an approximated case,
we evaluate the error of our algorithm with approximated iterations.
Since the em algorithm has local minimizers in general,
it is essential to cover the convergence to a local minimizer.
Since our algorithm has the acceleration parameter,
our application to the minimum divergence gives a generalization of the em algorithm.
Also, our algorithm can be applied to 
the maximization of the projected divergence to an exponential family by changing an element of a mixture family.

In addition, our algorithm has various applications that were not discussed in the preceding  
study \cite{RISB}.
In channel coding,
the decoding block error probability goes to zero exponentially 
under the proper random coding when the transmission rate is smaller than the capacity \cite{Gallager}.
Also, the probability of correct decoding goes to zero exponentially
when the transmission rate is greater than the capacity \cite{Arimoto2}.
These exponential rates are written with the optimization of the so-called Gallager function.
Recently, the paper \cite{H15} showed that the Gallager function 
can be written as the minimization of the R\'{e}nyi divergence.
Using this fact, we apply our method to these optimizations.
Further, we apply our algorithm to the capacity of a wiretap channel.
In addition, since our problem setting allows a general mixture family as the range of input,
we apply the channel capacity with cost constraint.
Also, 
we point out that the calculation of the commitment capacity is given as the 
minimization of the divergence between a mixture family and an exponential family.
Hence, we discuss this application as well.

The remaining part of this paper is organized as follows.
Section \ref{setup} formulates our minimization problem
for a general mixture family.
Then, we proposed several algorithms to solve the minimization problem.
We derive various convergence theorems including the case with 
approximated iterations.
The remaining sections apply our algorithm to various examples.
In these sections, examples of objective functions are discussed.
Section \ref{S4} applies our algorithm to various information theoretical problems.
Then, 
Section \ref{S5} applies our algorithm 
to the minimum divergence between a mixture family and an exponential family. 
Section \ref{S6} applies our algorithm to the commitment capacity.
Section \ref{S7} applies our algorithm to the maximization of the projected divergence to an exponential family by changing an element of a mixture family.
Section \ref{S7B} applies our algorithm to 
information bottleneck,
which is a powerful method for machine learning.
Appendices are devoted to the proofs of the theorems presented in Section \ref{setup}.

\section{General setting}\Label{setup}
\subsection{Algorithm with exact iteration}\Label{S2-1}
We consider a finite sample space ${\cal X}$
and focus on the set ${\cal P}({\cal X})$ of distributions whose support is ${\cal X}$.
Using $k$ linearly independent functions $f_1, \ldots, f_k$ on ${\cal X}$
and constants $a=(a_1, \ldots, a_k)$, 
we define the mixture family ${\cal M}_a$ as follows
\begin{align}
{\cal M}_a:= \{P \in {\cal P}({\cal X})| P[f_i]=a_i \hbox{ for } i=1, \ldots, k
\},\Label{MDP}
\end{align}
where $P[f]:= \sum_{x \in {\cal X}} P(x)f(x)$.
%If we have no constraint, it is written as ${\cal M}$.
We add additional $l-k$ linearly independent functions $f_{k+1}, \ldots
f_l$ and $|{\cal X}|=l+1$
such that
the $l$ functions $f_{1}, \ldots, f_l$ are linearly independent.
Then, the distribution $P$ can be parameterized by 
the mixture parameter $\eta=(\eta_1, \ldots, \eta_l)$ as
$ \eta_i= P[f_i]$.
That is, the above distribution is denoted by $P_\eta$.
Then, we denote the $e$-projection of $P$ to ${\cal M}_a$
by 
$\Gamma^{(e)}_{{\cal M}_a}[P]$.
That is, $\Gamma^{(e)}_{{\cal M}_a}[P]$ is defined as follows \cite{Amari1,Amari}.
\begin{align}
\Gamma^{(e)}_{{\cal M}_a}[P]:=
\argmin_{ Q \in {\cal M}_a}
D(Q\|P), \Label{Mix}
\end{align}
where the Kullback-Leibler divergence 
$D(Q\|P)$ is defined as
\begin{align}
D(Q\|P):= \sum_{x \in {\cal X}}Q(x) (\log Q(x)-
\log P(x)).
\end{align}
Using the $e$-projection, we have the following equation
for an element of $Q \in {\cal M}_a$, which is often called
Pythagorean theorem.
\begin{align}
D(Q\|P)=
D(Q\|\Gamma^{(e)}_{{\cal M}_a}[P])
+ D(\Gamma^{(e)}_{{\cal M}_a}[P]\|P).\Label{NNP}
\end{align}

Given a continuous function $\Psi$ from ${\cal M}_a$ to 
the set of functions on ${\cal X}$, 
we consider the minimization
$\min_{P \in {\cal M}_a} {\cal G}(P)$;
\begin{align}
{\cal G}(P):= \sum_{x \in {\cal X}} P(x) \Psi[P](x).
\Label{ZNV}
\end{align}
This paper aims to find
\begin{align}
\overline{{\cal G}}(a):=\min_{P \in {\cal M}_a} {\cal G}(P), \quad
P_{*,a}:=\argmin_{P \in {\cal M}_a} {\cal G}(P).
\Label{NM6}
\end{align}

For this aim, we propose an iterative algorithm 
based on a positive real number $\gamma>0$.
Since the above formulation \eqref{ZNV} is very general, 
we can choose the function $\Psi$
dependently on our objective function.
That is, different choices of $\Psi$ lead to different 
objective functions. 

For a distribution $Q \in {\cal P}({\cal X})$,
we define the distribution ${\cal F}_3[Q] $ as
\begin{align}
{\cal F}_3[Q](x):= \frac{1}{\kappa[Q]}Q(x)\exp(  -\frac{1}{\gamma} \Psi[Q](x)),\Label{VU8}
\end{align}
where $\kappa[Q]$ is the normalization factor
$\sum_{x \in {\cal X}} Q(x)\exp(  -\frac{1}{\gamma} \Psi[Q](x))$.
Then, %generalizing the algorithm by \cite{RISB}, 
depending on $\gamma>0$,
we propose Algorithm \ref{AL1}.
When the calculation of $ \Psi[P]$ and 
the $e$-projection is feasible, Algorithm \ref{AL1} is feasible.

\begin{algorithm}
\caption{Minimization of ${\cal G}(P)$}
\Label{AL1}
\begin{algorithmic}
\STATE {As inputs, we prepare 
the function $\Psi$,
$l$ linearly independent functions 
$f_1, \ldots, f_{l}$,
constraints $a_1, \ldots,a_k $,
a positive number $\gamma>0$, and
the initial value 
$P^{(1)} \in \mathcal{M}_a$;} 
\REPEAT 
%\STATE {Set $\hat{\theta}=\hat{\theta}^{(t)} \in \Theta_{\mathcal{M}}$;} 
\STATE Calculate $P^{(t+1)}:=\Gamma^{(e)}_{{\cal M}_a}[{\cal F}_3[P^{(t)}]]
$;
\UNTIL{convergence. We denote the convergent by $P^{(\infty)}$.
The convergence of this algorithm is guaranteed by Theorem \ref{TTH1}.
} 
\STATE{Output $P^{(\infty)}$. }
\end{algorithmic}
\end{algorithm}

\if0
In fact, the condition (A2) is closely related to the convexity of ${\cal G}(P)$ as follows.
\begin{lemma}
When 
\begin{align}
\sum_{x \in {\cal X}} P(x) (\Psi[P](x)- \Psi[Q](x)) \ge 0
\Label{XMZ2}
\end{align}
holds for $P,Q \in {\cal M}_a$, 
the map $ P \mapsto {\cal G}(P)$ is convex.
\end{lemma}
\begin{proof}
\begin{align}
&\lambda {\cal G}(P)+(1-\lambda) {\cal G}(Q)
-\lambda {\cal G}(\lambda P+(1-\lambda)Q)\\
=&\lambda \sum_{x \in {\cal X}} P(x) (\Psi[P](x)- \Psi[\lambda P+(1-\lambda)Q](x))\\
&+(1-\lambda) \sum_{x \in {\cal X}} Q(x) (\Psi[Q](x)- \Psi[\lambda P+(1-\lambda)Q](x))\\
 \ge & 0.
\end{align}
\end{proof}
\fi

Indeed, Algorithm \ref{AL1} is characterized as the iterative minimization of 
the following two-variable function, i.e., the extended objective function;
\begin{align}
J_\gamma(P,Q):=\gamma D(P\|Q)+\sum_{x \in {\cal X}} P(x) \Psi[Q](x).
\Label{VUI}
\end{align}
To see this fact, %consider an iterative algorithm, 
we define
\begin{align}
{\cal F}_1[P]  := \argmin_{Q \in {\cal M}_a}  J _\gamma(P,Q) ,\quad
{\cal F}_2[Q]  := \argmin_{P \in {\cal M}_a}  J _\gamma(P,Q) .
\Label{VUI2}
\end{align}

Then, 
${\cal F}_2[Q]$ is calculated as follows.

\begin{lemma}\Label{L1}
Under the above definitions, for any positive value $\gamma >0$,
we have ${\cal F}_2[Q] =\Gamma^{(e)}_{{\cal M}_a}[{\cal F}_3[Q]] $, i.e., 
\begin{align}
\min_{P \in {\cal M}_a}  J_\gamma(P,Q)&=
J _\gamma(\Gamma^{(e)}_{{\cal M}_a}[{\cal F}_3[Q]],Q) \nonumber \\
&=
\gamma D(\Gamma^{(e)}_{{\cal M}_a}[{\cal F}_3[Q]]\|{\cal F}_3[Q])
- \gamma \log \kappa[Q]  ,\Label{XMY} \\
 J _\gamma(P,Q)
 &=\min_{P' \in {\cal M}_a}  J _\gamma(P',Q)
+\gamma D(P\| \Gamma^{(e)}_{{\cal M}_a}[{\cal F}_3[Q]]) \Label{XMY2UU}  \\
&=J _\gamma(\Gamma^{(e)}_{{\cal M}_a}[{\cal F}_3[Q]],Q)
+\gamma D(P\| \Gamma^{(e)}_{{\cal M}_a}[{\cal F}_3[Q]]).
\Label{XMY2} 
\end{align}
\end{lemma}
\begin{proof}
We have the following relations.
\begin{align}
&J _\gamma(P,Q)
=\gamma \sum_{x \in {\cal X}} P(x) (\log P(x)- \log Q(x) + \frac{1}{\gamma} \Psi[Q](x)) \nonumber\\
=&\gamma \sum_{x \in {\cal X}} P(x) (\log P(x)- \log {\cal F}_3[Q](x)- \log \kappa[Q]) \nonumber\\
=&\gamma D(P\| {\cal F}_3[Q])-\gamma \log \kappa[Q] \nonumber\\
=&\gamma D(P\| \Gamma^{(e)}_{{\cal M}_a}[{\cal F}_3[Q]])
+\gamma D(\Gamma^{(e)}_{{\cal M}_a}[{\cal F}_3[Q]]\|{\cal F}_3[Q])
- \gamma\log \kappa[Q] ,\Label{ASS4}
\end{align}
where the final equation follows from 
\eqref{NNP}.
Then, the minimum is given as \eqref{XMY}, and it is realized with 
$\Gamma^{(e)}_{{\cal M}_a}[{\cal F}_3[Q]]$.

Applying \eqref{XMY} to the final line of \eqref{ASS4},
we obtain \eqref{XMY2UU}.
Since the minimum in \eqref{XMY2UU} is realized when 
$P'=\Gamma^{(e)}_{{\cal M}_a}[{\cal F}_3[Q]]$, 
we obtain \eqref{XMY2}.
\end{proof}

We calculate ${\cal F}_1[Q]$.
For this aim, we define
\begin{align}
D_{\Psi}(P \|Q):= 
\sum_{x \in {\cal X}} P(x) (\Psi[P](x)- \Psi[ Q ](x)).\Label{TM9}
\end{align}

\begin{lemma}\Label{L2}
Assume that two distributions $P,Q \in {\cal M}_a$ satisfy the following condition,
\begin{align}
D_{\Psi}(P \|Q) 
%\sum_{x \in {\cal X}} P(x) (\Psi[P](x)- \Psi[Q](x))
\le \gamma D(P\|Q).
\Label{BK1+}
\end{align}
Then, we have ${\cal F}_1[P] =P$, i.e., 
\begin{align}
J _\gamma(P,Q)\ge J _\gamma(P,P) .
\end{align}
\end{lemma}
\begin{proof}
Eq. \eqref{BK1+} guarantees that
\begin{align}
&J _\gamma(P,Q)-J _\gamma(P,P) \nonumber\\
=&
\gamma D(P\|Q)-\sum_{x \in {\cal X}} P(x) 
(\Psi[P](x)- \Psi[Q](x))
\ge 0 \Label{XMY5}.
\end{align}
\end{proof}

\begin{remark}\Label{RR1}
The preceding study \cite{RISB} discussed the minimization of 
a function defined over the set of density matrices, i.e., the set of quantum states.
When the function is given as a function only of the diagonal part,
the function is given as a function of probability distribution composed of the diagonal part.
That is, the preceding study \cite{RISB} covers the case when 
the function is optimized over a set of probability distributions
as a special case.
The obtained result of this paper covers the case when 
the function is optimized over a mixture family.
That is, the preceding study \cite{RISB} does not consider the case with linear constraints.
In this sense, the obtained result of this paper 
generalizes the above special case of the result of \cite{RISB},
and Algorithm \ref{AL1} is a generalization of the algorithm given in \cite{RISB}.

Lemma 3.2 \cite{RISB} is composed of several statements.
The combination of Lemmas \ref{L1} and \ref{L2} 
is a generalization of the above special case of \cite[Lemma 3.2]{RISB}.
That is, the classical restriction of \cite[Lemma 3.2]{RISB} is equivalent to 
the combination of Lemmas \ref{L1} and \ref{L2} 
without linear constraints.
%In the latter part, we refer the preceding study \cite{RISB} several times.
%At that point, the preceding study \cite{RISB} already derived the same result without 
\end{remark}

Therefore, when all pairs $(P^{(t+1)},P^{(t)})$ satisfy \eqref{BK1+}, 
the relations 
\begin{align}
{\cal G}(P^{(t)})=J _\gamma(P^{(t)},P^{(t)})\ge
J _\gamma(P^{(t+1)},P^{(t)})
\ge J _\gamma(P^{(t+1)},P^{(t+1)})= {\cal G}(P^{(t+1)})
\Label{SAC}
\end{align}
hold under Algorithm \ref{AL1}. % and the condition (A1).
In addition, we have the following theorem.

\begin{theorem}\Label{TTH1}
When all pairs $(P^{(t+1)},P^{(t)})$ satisfy \eqref{BK1+},
i.e., the positive number $\gamma$ is sufficiently large,
Algorithm \ref{AL1}
converges to a local minimum.
\end{theorem}

\begin{proof}
Since $\{{\cal G}(P^{(t)}) \}$
is monotonically decreasing for $t$,
we have 
\begin{align}
\lim_{n \to \infty}
{\cal G}(P^{(t)})-{\cal G}(P^{(t+1)})
  =0\Label{AAS}.
\end{align}
Using \eqref{XMY2}, we have
  \begin{align}
&{\cal G}(P^{(t)})
=J_{\gamma}(P^{(t)},P^{(t)})  \nonumber \\
=& \gamma 
 D(P^{(t)}\| P^{(t+1)})
+
J_{\gamma}( P^{(t+1)}, P^{(t)}) \nonumber \\
\ge 
& \gamma 
 D(P^{(t)}\| P^{(t+1)})
+
{\cal G}( P^{(t+1)}) .
\end{align}
Thus, we have
\begin{align}
 \gamma  D(P^{(t)}\| P^{(t+1)})
\le 
{\cal G}(P^{(t)})-{\cal G}(P^{(t+1)})
  \Label{AAS3}.
\end{align}
Since
due to \eqref{AAS} and \eqref{AAS3},
the sequence $\{{\cal G}(P^{(t)})\}$ is a Cauchy sequence, it converges. 
\end{proof}

To discuss the details of Algorithm \ref{AL1}, 
%To discuss Algorithm \ref{AL1}, 
we focus on the $\delta$-neighborhood $U(P^{0},\delta)$ of $P^{0}$ defined as
\begin{align}
U(P^{0},\delta):=\{
P \in {\cal M}_a | D(P^{0} \|P)\le \delta \}.\Label{VUI7}
\end{align}
In particular, we denote ${\cal M}_a$ by $U(P^{0},\infty)$.
%When $\delta=\infty$, 
Then, we address the following conditions for %a pair of $P^{0} $ and $$
the $\delta$-neighborhood $U(P^{0},\delta)$ of $P^{0}$;
\begin{description}
\item[(A0)]
Any distribution $Q \in U(P^{0},\delta)$ satisfies the inequality
\begin{align}
{\cal G}({\cal F}_2[Q]) 
%J _\gamma({\cal F}_2[Q],Q)
\ge {\cal G}(P^{0}).\Label{BK-1}
\end{align}
\if0
\item[(A0')]
Any distribution $Q \in U(P^{0},\delta)$ satisfies the inequality
\begin{align}
%{\cal G}(Q) 
{\cal G}(Q)
\ge {\cal G}(P^{0}).\Label{BK-2}
\end{align}
\fi
\item[(A1)]
Any distribution $Q \in U(P^{0},\delta)$ satisfy
\begin{align}
D_{\Psi}({\cal F}_2[Q] \|Q) \le \gamma D({\cal F}_2[Q] \| Q)
\Label{BK1}.
\end{align}
\item[(A2)]
Any distribution $Q \in U(P^{0},\delta)$ satisfies
\begin{align}
D_{\Psi}(P^{0} \|Q) \ge 0.
\Label{XMZ}
\end{align}
\item[(A3)]
There exists a positive number $\beta>0$ such that
any distribution $Q \in U(P^{0},\delta)$ satisfies
\begin{align}
D_{\Psi}(P^{0} \|Q)
=\sum_{x \in {\cal X}} P^{0}(x) (\Psi[P^{0}](x)- \Psi[Q](x)) 
\ge \beta
D(P^{0} \| Q).
\Label{CAU}
\end{align}
\end{description}
The condition (A3) is a stronger version of (A2).
\if0
When (A1) holds, ${\cal G}(Q) \ge {\cal G}({\cal F}_2[Q]) $, i.e., 
(A0) implies (A0'). Hence, 
the condition (A0') is a weaker version of (A0).
\fi

However, the convergence to the global minimum is not guaranteed.
As a generalization of \cite[Theorem 3.3]{RISB},  
the following theorem discusses the convergence to the global minimum 
and the convergence speed.

\begin{theorem}\Label{TH1}
Assume that 
the $\delta$-neighborhood $U(P^{0},\delta)$ of $P^{0}$
satisfies the conditions %(A0'), 
(A1) and (A2) with $\gamma$, and $P^{(1)} \in U(P^{0},\delta)$.
Then, 
Algorithm \ref{AL1} with $t_0$ iterations
has one of the following two behaviors.
\begin{description}
\item[(i)]
There exists an integer $t_1 \le t_0+1$ such that 
\begin{align}
{\cal G}(P^{(t_1)}) < {\cal G}(P^{0}).
\end{align}
\item[(ii)]
Algorithm \ref{AL1} satisfies the conditions
$\{P^{(t)}\}_{t=1}^{t_0+1} 
\subset U(P^{0},\delta)$ and
\begin{align}
{\cal G}(P^{(t_0+1)})
-{\cal G}(P^{0})
\le 
\frac{\gamma D(P^{0}\| P^{(1)}) }{t_0}.\Label{XME}
\end{align}
\end{description}
When the condition (A0) holds additionally, %instead of the conditions (A0'),
Algorithm \ref{AL1} with $t_0$ iterations satisfies (ii).
\end{theorem}

The above theorem is shown in Appendix \ref{S3-2}.
Now, we choose an element $P^* \in {\cal M}_a$ to satisfy  
${\cal G}(P^*)=\min_{P \in {\cal M}_a} {\cal G}(P)$.
Then, 
the condition (A0) holds with 
$U(P^*,\infty)={\cal M}_a$ and the choice $P^0=P^*$.
When the conditions (A1) and (A2) hold with
$U(P^*,\infty)={\cal M}_a$ and the choice $P^0=P^*$,
Theorem \ref{TH1} guarantees the convergence to the minimizer $P^*$ in Algorithm \ref{AL1}.
Although Theorem \ref{TH1} requires the conditions (A1) and (A2), the condition (A2) is essential due to the following reason.
When we choose $\gamma>0$ to be sufficiently large,
%and $\delta\neq \infty$,
the condition (A1) holds with
the $\delta$-neighborhood $U(P^*,\delta)$ of $P^*$ 
because $U(P^*,\delta)$ is a compact set.
Hence, 
%Theorem \ref{TH1} guarantees the convergence of Algorithm \ref{AL1}.
%That is, 
it is essential to check the condition (A2) for Theorem \ref{TH1}.

However, as seen in \eqref{XME}, a larger $\gamma$ makes the convergence speed slower.
Therefore, it is important to choose $\gamma$ to be small under the condition (A1).
Practically, it is better to change $\gamma$ to be smaller when the point $P^{(t)}$ is closer to 
the minimizer $P^*$.
In fact, 
as a generalization of \cite[Proposition 3.6]{RISB},  
we have the following exponential convergence under a stronger condition dependently of $\gamma$.
In this sense, the parameter is called an acceleration parameter \cite[Remark 3.4]{RISB}.

\begin{theorem}\Label{TH2}
Assume that 
the $\delta$-neighborhood $U(P^{0},\delta)$ of $P^{0}$
satisfies the conditions %(A0'), 
(A1) and (A3) with $\gamma$, and $P^{(1)} \in U(P^{0},\delta)$.
Then, 
Algorithm \ref{AL1} with $t_0$ iterations
has one of the following two behaviors.
\begin{description}
\item[(i)]
There exists an integer $t_1 \le t_0+1$ such that 
\begin{align}
{\cal G}(P^{(t_1)}) < {\cal G}(P^{0}).
\end{align}
\item[(ii)]
Algorithm \ref{AL1} satisfies the conditions
$\{P^{(t)}\}_{t=1}^{t_0+1} 
\subset U(P^{0},\delta)$ and
\begin{align}
 {\cal G}(P^{(t_0+1)})-{\cal G}(P^{0}) 
\le (1-\frac{\beta}{\gamma})^{t_0} D(P^{0} \| P^{(1)}).\Label{CAU2}
\end{align}
\end{description}
When the condition (A0) holds additionally, %instead of the conditions (A0'),
Algorithm \ref{AL1} with $t_0$ iterations satisfies (ii).
\if0
the $\delta$-neighborhood $U(P^{0},\delta)$ of $P^{0}$
satisfies the conditions (A0), (A1), and (A3) with $\gamma$ and $\beta$, and $P^{(1)} \in U(P^{0},\delta)$.
Then, Algorithm \ref{AL1} satisfies the conditions
$\{P^{(t)}\} \subset U(P^{0},\delta)$ and
\begin{align}
 {\cal G}(P^{(t+1)})-{\cal G}(P^*) 
\le (1-\frac{\beta}{\gamma})^t D(P^* \| P^{(1)}).\Label{CAU2}
\end{align}
\fi
\end{theorem}

The above theorem is shown in Appendix \ref{S3-3}.
Next, we consider the case when there are several local minimizers $P^*_1, \ldots, P^*_{n}
\in {\cal M}_a$
while the true minimizer is $P^*$.
These local minimizers are characterized by the following corollary,
which is shown in Appendix \ref{S3-2} as a corollary of Theorem \ref{TH1}.

\begin{corollary}\Label{Cor1}
\begin{align}
D_{\Psi}(P^* \|P^*_i)= 
\sum_{x \in {\cal X}}P^*(x) 
(\Psi[P^*](x)-\Psi[P^*_i](x))= 
{\cal G}(P^*)-{\cal G}(P^*_i)<0.\Label{Cor2}
\end{align}
\end{corollary}

Hence, if there exist local minimizers, 
the condition (A2) does not hold with 
$U(P^*,\infty)={\cal M}_a$ and the choice $P^0=P^*$.
In this case, when 
the $\delta$-neighborhood $U(P^*_i,\delta)$ of $P^*_i$
satisfies the conditions (A0), (A1), and (A2),
Algorithm \ref{AL1} converges to the local minimizer $P^*_i$ with the speed \eqref{XME}
except for the case (i).
Since $P^*_i$ is a local minimizer,
the $\delta$-neighborhood $U(P^*_i,\delta)$ of $P^*_i$
satisfies the conditions (A0) and (A1) with sufficiently small $\delta>0$.
When the following condition (A4) holds,
as shown below, the $\delta$-neighborhood $U(P^*_i,\delta)$ of $P^*_i$ satisfies the condition (A2)
with sufficiently small $\delta>0$.
That is, when the initial point belongs to the $\delta$-neighborhood $U(P^*_i,\delta)$,
Algorithm \ref{AL1} converges to $P^*_i$.
%To consider such a case, instead of (A2), we assume the following condition.
\begin{description}
\item[(A4)]
The function $\eta \mapsto \Psi[P_\eta](x)$ is differentiable, and
the relation 
\begin{align}
\sum_{x \in {\cal X}}P_{\eta} (x) 
\Big(\frac{\partial }{\partial \eta_i}\Psi[P_\eta](x)\Big)= 0
\end{align}
holds for $i=k+1, \ldots, l$ and $P_{\eta} \in {\cal M}_a$.
\end{description}

\begin{lemma}\Label{L6}
We consider the following two conditions for a convex subset ${\cal K} \subset {\cal M}_a$.
\begin{description}
\item[(B1)]
The relation
\begin{align}
D_{\Psi}(P \|Q)=\sum_{x \in {\cal X}} P(x) (\Psi[P](x)-\Psi[Q](x)) \ge 0
\end{align}
holds for $P,Q \in {\cal K}$.
\item[(B2)]
${\cal G}(P)$ is convex for the mixture parameter in ${\cal K}$.
\end{description}
The condition (B1) implies the condition (B2).
In addition, when the condition (A4) holds, 
the condition (B2) implies the condition (B1).
\end{lemma}

We consider two kinds of mixture parameters.
These parametrizations can be converted to each other via
affine conversion, which preserves the convexity.
Therefore, the condition (B2) does not depend on the choice of mixture parameter.

When the function $\eta \mapsto \Psi[P_\eta](x)$ is twice-differentiable, and
the Hessian of ${\cal G}(P_\eta)$
is strictly positive semi-definite at a local minimizer $ P^*_i$,
this function is convex in the $\delta$-neighborhood $U(P^*_i,\delta)$ of $P^*_i$ 
with a sufficiently small $\delta>0$ 
because the Hessian of ${\cal G}(P_\eta)$ is strictly positive semi-definite in the neighborhood due to the continuity.

Then, Lemma \ref{L6} guarantees the condition (A2) for the $\delta$-neighborhood $U(P^*_i,\delta)$.
Algorithm \ref{AL1} converges to the local minimizer $P^*_i$ with the speed \eqref{XME}
except for the case (i).
The mathematical symbols introduced in Section \ref{S2-1}
is summarized in Table \ref{symbols}

\begin{table}[t]
\caption{List of mathematical symbols for Section \ref{S2-1}}
\label{symbols}
\begin{center}
\begin{tabular}{|l|l|l|}
\hline
Symbol& Description & Eq. number  \\
\hline
${\cal P}(\cX)$ & Set of probability distributions over $\cX$&    \\
\hline
${\cal M}_a$ & Mixture family & \eqref{MDP} \\
\hline
$\Gamma^{(e)}_{{\cal M}_a}$ & $e$-projection to ${\cal M}_a$
& \eqref{Mix} \\
\hline
${\cal G}(P)$ & Objective function & \eqref{ZNV} \\
\hline
$\Psi[P]$ & Functional of $P$ used for objective function & \eqref{ZNV} \\
\hline
$\overline{{\cal G}}(a)$ & Minimum value of ${\cal G}(P)$ & \eqref{NM6} \\
\hline
$P_{*,a}$ & Minimizer of ${\cal G}(P)$ & \eqref{NM6} \\
\hline
${\cal F}_3[Q]$ & Functional of $Q$ & \eqref{VU8} \\
\hline
$J_\gamma(P,Q)$ & Extended objective function & \eqref{VUI} \\
\hline
${\cal F}_1[P]$ & Minimizer of $J_\gamma(P,Q)$ for second argument  & \eqref{VUI2} \\
\hline
${\cal F}_2[Q]$ & Minimizer of $J_\gamma(P,Q)$ for first argument  & \eqref{VUI2} \\
\hline
$D_{\Psi}(P \|Q)$ & Function of $P$ and $Q$ related to $\Psi$ & \eqref{TM9} \\
\hline
$U(P^{0},\delta)$ & $\delta$-neighborhood of $P^0$  & \eqref{VUI7} \\
\hline
\end{tabular}
\end{center}
\end{table}

\begin{proofof}{Lemma \ref{L6}}
Assume the condition (B1). Then, for $\lambda \in [0,1]$, we have
\begin{align}
\varphi(\lambda):=&\lambda {\cal G}(P)+(1-\lambda) {\cal G}(Q)
- {\cal G}(\lambda P+(1-\lambda)Q)\nonumber \\
=&\lambda \sum_{x \in {\cal X}} P(x) (\Psi[P](x)- \Psi[\lambda P+(1-\lambda)Q](x))
\nonumber \\
&+(1-\lambda) \sum_{x \in {\cal X}} Q(x) (\Psi[Q](x)- \Psi[\lambda P+(1-\lambda)Q](x))\nonumber \\
 \ge & 0,
\end{align}
which implies (B2).

Assume the conditions (A4) and (B2).
Since $\varphi(\lambda)\ge 0$ for $\lambda \in [0,1]$,
we have
\begin{align}
0\le &\frac{d \varphi(\lambda)}{d \lambda}|_{\lambda=0} \nonumber \\
=&{\cal G}(P)- {\cal G}(Q)
-\sum_{x \in {\cal X}} (P(x)-Q(x))\Psi[Q](x)\nonumber \\
&-\sum_{x \in {\cal X}} Q(x) 
\frac{d \Psi[\lambda P+(1-\lambda)Q](x)}{d\lambda}|_{\lambda=0}\nonumber \\
\stackrel{(a)}{=}&{\cal G}(P)- {\cal G}(Q)
-\sum_{x \in {\cal X}} (P(x)-Q(x))\Psi[Q](x),
\end{align}
which implies (B1), where $(a)$ follows from 
the condition (A4).
\end{proofof}

\begin{remark}
The preceding study \cite[Theorem 3.3 and Proposition 3.6]{RISB} 
consider similar statements as Theorems \ref{TH1} and \ref{TH2}.
As mentioned in Remark \ref{RR1},
the preceding study \cite{RISB}
covers the case when
${\cal M}_a$ is given as ${\cal P}({\cal X})$,
and does not cover the case with a general mixture family ${\cal M}_a$.
In addition, the preceding study \cite[Theorem 3.3 and Proposition 3.6]{RISB} covers only the case 
when $P^0$ and $U(P^{0},\delta)$ are
$P^*$ and $ {\cal P}({\cal X})$, respectively.
That is, the preceding study does not cover the case with local minimizers.
In this sense, 
Theorems \ref{TH1} and \ref{TH2} are more general under the classical setting.
\end{remark}

\if0
\begin{remark}
The reference \cite{RISB} considers the same problem setting
with the quantum case.
Their case covers the case when
${\cal M}_a$ is given as ${\cal P}({\cal X})$,
and does not cover the case with a general mixture family ${\cal M}_a$.
To address the case with cost constraint,
this extension is essential.
The ideas of Theorems \ref{TH1} and \ref{TH2} are quite similar to 
\cite[Theorem 3.3 and Proposition 3.6]{RISB}.
However, these preceding studies consider the case when 
$P^0$ and $U(P^{0},\delta)$ are
$P^*$ and $ {\cal P}({\cal X})$, respectively.
That is, they do not cover the case with local minimizers.
\end{remark}
\fi

\subsection{Algorithm with approximated iteration}\Label{S2-2}
%=&J _\gamma(P^{(t+1)},P^{(t)})\\
In general, it is not so easy to calculate 
the $e$-projection $\Gamma^{(e)}_{{\cal M}_a}({\cal F}_3[Q])$.
We consider the case when it is approximately calculated.
There are two methods to calculate the $e$-projection.
One is the method based on the minimization in the given mixture family,
and the other is the method based on the minimization in the exponential family orthogonal to the mixture family.
In the first method, 
the $e$-projection $\Gamma^{(e)}_{{\cal M}_a}({\cal F}_3[Q])$
is the minimizer of the following minimization;
\begin{align}
\min_{P \in {\cal M}_a} D(P\| {\cal F}_3[Q]).\Label{ZMY}
\end{align}

\if0
To characterize the $e$-projection $\Gamma^{(e)}_{{\cal M}_a}$
of ${\cal F}_3[Q]$,
we define the exponential family
\begin{align}
Q_{\theta}(x):= {\cal F}_3[Q](x) e^{\sum_{j=1}^k \theta^j f_j(x)- \phi[Q](\theta)}, \Label{E-fami}
\end{align}
where
the functions $f_j$ correspond
to the linear independent functions used in equation \eqref{MDP}, and 
$ \phi[Q](\theta):= \sum_{x \in {\cal X}}{\cal F}_3[Q](x) e^{\sum_{j=1}^k \theta^j  f_j(x)}$.
\fi
To describe the second method, we define the exponential family
\begin{align}
Q_{\theta}(x):= {\cal F}_3[Q](x) e^{\sum_{j=1}^k \theta^j f_j(x)- \phi[Q](\theta)}, \Label{E-fami}
\end{align}
where
\begin{align}
 \phi[Q](\theta):=
 \log \sum_{x \in {\cal X}}{\cal F}_3[Q](x) e^{\sum_{j=1}^k \theta^j  f_j(x)}.\Label{E-fami2}
\end{align}
The projected element $\Gamma^{(e)}_{{\cal M}_a}[{\cal F}_3[Q]]$
is the unique element of the intersection
$ \{Q_\theta\}\cap {\cal M}_a$.
For example, for this fact, see \cite[Lemma 3]{Bregman-em}.
Then, the $e$-projection $\Gamma^{(e)}_{{\cal M}_a}({\cal F}_3[Q])$
is given as the solution of the following equation;
\begin{align}
\frac{\partial \phi[Q]}{\partial \theta^j }(\theta)
=\sum_{x\in {\cal X}}Q_\theta(x) f_j(x)= a_j\Label{ZMX}
\end{align}
for $j=1, \ldots, k$.
The solution of \eqref{ZMX} is given as the minimizer of the following minimization;
\begin{align}
\min_{\theta \in \mathbb{R}^k}
\phi[Q](\theta)- \sum_{j=1}^k \theta^j a_j.\Label{ZMX2}
\end{align}

We discuss the precision of our algorithm when each step in
the above minimization has a certain error.
%We discuss the precision of our algorithm when
%each step has a certain error in the above minimization.

\if0
To discuss the first method with error, we consider Algorithm \ref{AL2} instead of Algorithm \ref{AL1}.
\begin{algorithm}
\caption{Minimization of ${\cal G}(P)$ with $\epsilon$ error in \eqref{ZMY}}
\Label{AL2}
\begin{algorithmic}
\STATE {Choose the initial value $P^{(1)} \in \mathcal{M}$;} 
\REPEAT 
%\STATE {Set $\hat{\theta}=\hat{\theta}^{(t)} \in \Theta_{\mathcal{M}}$;} 
\STATE Calculate $P^{(t+1)}$ to satisfy
\begin{align}
D(P^{(t+1)}\| {\cal F}_3[P^{(t)}])
\le \min \Big( D(P^{(t)}\| {\cal F}_3[P^{(t)}]),
 \min_{P \in {\cal M}_a} D(P\| {\cal F}_3[P^{(t)}]) +\epsilon \Big);
\end{align}
\UNTIL{convergence.} 
\end{algorithmic}
\end{algorithm}
\fi
\if0
\begin{theorem}
When \eqref{XMZ} and \eqref{BK1} hold,
we have
\begin{align}
{\cal G}(P^{(k+1)})
-{\cal G}(P^*)
\le 
\gamma \frac{D(P^*\| P^{(1)}) }{k} + 2 \kappa \sqrt{D(P^*\| P^{(1)}) \epsilon}
+(\kappa+1)\epsilon. \Label{XZW}
\end{align}
\end{theorem}
\fi

However, the first method requires the minimization with the same number of parameters
as the original minimization $\min_{P \in {\cal M}_a} {\cal G}(P)$.
Hence, it is better to employ the second method.
In fact, when ${\cal M}_a$ is given as a subset of ${\cal P}({\cal X})$ with one linear constraint,
the minimization \eqref{ZMX2} is written as a one-parameter convex minimization.
Since any one-parameter convex minimization is performed by the bisection method, which needs $O(-\log \epsilon)$ iterations \cite{BV}
to achieve a smaller error of the minimum of the objective function
than $\epsilon$,
the cost of this minimization is much smaller than that of the original minimization
$\min_{P \in {\cal M}_a} {\cal G}(P)$.
To consider an algorithm based on the minimization 
\eqref{ZMX2}, we assume that $\Psi$ is defined in ${\cal P}({\cal X})$.
In the multi-parameter case, we can use
the gradient method and 
the accelerated proximal gradient method \cite{BT,Nesterov,AT,Nesterov2,Nesterov3,Teboulle}.

\begin{algorithm}
\caption{Minimization of ${\cal G}(P)$ with an error in \eqref{ZMY}}
\Label{AL3}
\begin{algorithmic}
\STATE {As inputs, we prepare 
the function $\Psi$,
$l$ linearly independent functions 
$f_1, \ldots, f_{l}$,
constraints $a_1, \ldots,a_k $,
positive numbers $\gamma,\epsilon_1,\epsilon_2>0$, and
the initial value 
$P^{(1)} \in \mathcal{M}_a$;} 
%\STATE {Choose the initial value $P^{(1)} \in \mathcal{M}$;} 
\REPEAT 
%\STATE {Set $\hat{\theta}=\hat{\theta}^{(t)} \in \Theta_{\mathcal{M}}$;} 
\STATE 
Calculate the pair of $P^{(t+1)} \in {\cal M}_a$ and $\bar{P}^{(t+1)}= Q_\theta$ with $Q=P^{(t)}$ in \eqref{E-fami}
to satisfy
\begin{align}
\phi[\bar{P}^{(t)}](\theta)- \sum_{j=1}^k \theta^j a_j
&\le 
\min_{\theta' \in \mathbb{R}^k}
\phi[\bar{P}^{(t)}](\theta')- \sum_{j=1}^k {\theta'}^j a_j
+\epsilon_1 \Label{AMG} \\
D(\bar{P}^{(t+1)}\| P^{(t+1)} ) &\le \epsilon_2\Label{NXP} .
%\\
%{\cal G}(P^{(t+1)}) & \le  {\cal G}(\bar{P}^{(t+1)}) + \epsilon_3  \Label{NXP2}.
\end{align}
\UNTIL{$t=t_1-1$.} 
\STATE {\bf final step:}\quad 
We output the final estimate $P_f^{(t_1)} :=P^{(t_2)} \in \mathcal{M}$
by using  $t_2:= \argmin_{t=2, \ldots, t_1} 
{\cal G}(P^{(t)})- \gamma D(P^{(t)} \| \bar{P}^{(t)})$.
\end{algorithmic}
\end{algorithm}

To consider the convergence of Algorithm \ref{AL3}, we extend
the conditions (A1) and (A2).
For this aim, we focus on the $\delta$-neighborhood $\bar{U}(P^{0},\delta)$ of $P^{0} \in {\cal M}_a$ defined as
\begin{align}
\bar{U}(P^{0},\delta):=\{
P \in {\cal P}({\cal X}) | D(P^{0} \|P)\le \delta \}.
\end{align}
Then, we introduce the following conditions for 
the $\delta$-neighborhood $\bar{U}(P^{0},\delta)$ of $P^{0}$ as follows.
\begin{description}
\if0
\item[(A0+)]
A distribution $Q \in \bar{U}(P^{0},\delta)$ satisfies
\eqref{BK-1}.
\item[(A0'+)]
A distribution $Q \in \bar{U}(P^{0},\delta)$ satisfies
\eqref{BK-2}.
\fi
\item[(A1+)]
%a distribution $Q \in \bar{U}(P^{0},\delta)$ satisfy
Any distribution $Q \in \bar{U}(P^{0},\delta)\cap {\cal M}_a
={U}(P^{0},\delta)$ 
satisfies the following condition with a positive real number $\epsilon_2 >0$.
When a distribution $P \in {\cal M}_a$ satisfies $D( P \|{\cal F}_2[Q])
\le \epsilon_2$, 
we have
%$Q \in \bar{U}(P^{0},\delta)$ satisfy
\begin{align}
\sum_{x \in {\cal X}} P(x) (\Psi[P](x)- \Psi[ Q ](x))
\le \gamma D(P \| Q).
\Label{BK12}
\end{align}
%\eqref{BK1}.
\item[(A2+)]
A distribution $Q \in \bar{U}(P^{0},\delta)$ satisfies \eqref{XMZ}.
\end{description}

The convergence of Algorithm \ref{AL3} is guaranteed in the following theorem.

\begin{theorem}\Label{TH8}
Assume that 
the $\delta$-neighborhood $\bar{U}(P^{0},\delta)$ of $P^{0}$
satisfies the conditions %(A0+), 
(A1+) and (A2+) with two positive real numbers $\gamma>0$, $\epsilon_2>0$, and $P^{(1)} \in U(P^{0},\delta)$.
Then, for a positive real number $\epsilon_1>0$,
Algorithm \ref{AL3} satisfies the conditions
%$\{P^{(t)}\} \subset U(P^{0},\delta)$ and
\begin{align}
D(\Gamma^{(e)}_{{\cal M}_a}[{\cal F}_3[\bar{P}^{(t)}]]\| \bar{P}^{(t+1)} ) & \le \epsilon_1 
\Label{XP8}\\
{\cal G}(P_f^{(t_1)})
-{\cal G}(P^*)
& \le 
 \frac{\gamma D(P^*\| P^{(1)}) }{t_1-1} + \epsilon_1%+\epsilon_3 
 +\gamma \epsilon_2.
 \Label{XZWN}
\end{align}
\end{theorem}

The above theorem is shown in Appendix \ref{S3-4}.
We discussed the convergences of Algorithms \ref{AL1} and \ref{AL3}
under several conditions.
When these conditions do not hold,
we cannot guarantee its global convergence but, the algorithms achieve a local minimum.
Hence, we need to repeat these algorithms by changing the initial value.
The mathematical symbols introduced in Section \ref{S2-2}
is summarized in Table \ref{symbols2}

\begin{table}[t]
\caption{List of mathematical symbols for Section \ref{S2-2}}
\label{symbols2}
\begin{center}
\begin{tabular}{|l|l|l|}
\hline
Symbol& Description & Eq. number  \\
\hline
$Q_{\theta}$ & Exponential family &
\eqref{E-fami} \\
\hline
$\phi[Q](\theta)$ & Potential function &\eqref{E-fami2}\\
\hline
\end{tabular}
\end{center}
\end{table}

\begin{remark}
To address the minimization with a cost constraint,
the paper \cite{YSM} added a linear penalty term to 
the objective function.
However, this method does not guarantee that the obtained result satisfies the required cost constraint.
Our method can be applied to any mixture family including the distribution family with cost constraint(s).
Hence, our method can be applied directly without the above modification
while we need to calculate the $e$-projection.
As explained in this subsection, this $e$-projection can be obtained with the convex minimization
whose number of variables is the number of the constraint to define the mixture family.
If the number of the constraints is not so large, still the $e$-projection is feasible.
\end{remark}

\subsection{Combination of the gradient method and the Algorithm \ref{AL1}}
%Combination of the gradient method and the Algorithm 1
\Label{S23}
Although we can use 
the gradient method to calculate \eqref{ZMX2} for a general mixture family ${\cal M}_a$,
in order to calculate $\overline{{\cal G}}(a):=\min_{P \in {\cal M}_a} {\cal G}(P)$
with $a \in \mathbb{R}^k$,
we propose another algorithm to combine
the gradient method and Algorithm \ref{AL1}.
This algorithm avoids the calculation of the $e$-projection
$\Gamma^{(e)}_{{\cal M}_a}$.
For simplicity, 
we assume that the function $\overline{{\cal G}}(a)$ is convex
and ${\cal M}_a$ is not empty,
and the aim is the calculation of $\overline{{\cal G}}(0)$.
In the following, we denote the expectation of the function $f$ 
under the distribution $P$ by $P[f]$.

Then, we consider the following functions by using 
Legendre transform; 
For $b=(b^1, \ldots, b^k)\in \mathbb{R}^k$ and 
$c=(c^1, \ldots, c^{l-k})\in \mathbb{R}^{l-k}$, we define
\begin{align}
{\cal G}_*(b,c)
:=& \sup_{P \in {\cal P}({\cal X})} \sum_{i=1}^k b^i P[f_i]
+\sum_{j=1}^{l-k} c^i P[f_{k+i}]-{\cal G}(P)  ,
\end{align} 
and
\begin{align}
\overline{{\cal G}}_*(b):=&{\cal G}_*(b,0)
=\sup_{P \in {\cal P}({\cal X})} \sum_{i=1}^k b^i P[f_i]-{\cal G}(P)  
=  \sup_{a \in \mathbb{R}^k} 
\sum_{i=1}^k b^i a_i -\overline{{\cal G}}(a).\Label{ASS}
\end{align} 
In the following, we consider the calculation of 
$\overline{{\cal G}}(0)$ by assuming that
the function $\eta \mapsto {\cal G}(P_\eta) $ is $C^2$-continuous and convex.
Since 
Legendre transform of $\overline{{\cal G}}_*(b)$ is 
$\overline{{\cal G}}(a)$ due to the convexity of $\overline{{\cal G}}(a)$,
we have %\begin{align}
$\sup_{b \in \mathbb{R}^k} 
\sum_{i=1}^k b^i a_i -\overline{{\cal G}}_*(b)
=\overline{{\cal G}}(a)$.
As a special case, 
we have
\begin{align}
-\inf_{b \in \mathbb{R}^k} 
\overline{{\cal G}}_*(b)
=\overline{{\cal G}}(0).\Label{ZMQ}
\end{align}
That is, when we find the minimizer
$b_*:=  \argmin_{a \in \mathbb{R}^k} 
\overline{{\cal G}}_*(b)$, we can calculate $\overline{{\cal G}}(0)$
as
%\begin{align}
$\overline{{\cal G}}(0)= 
-\sup_{P \in {\cal P}({\cal X})} \sum_{i=1}^k b_*^i P[f_i]-{\cal G}(P) 
=\inf_{P \in {\cal P}({\cal X})} {\cal G}(P)-\sum_{i=1}^k b_*^i P[f_i] $.
%\end{align}

To find it, we denote the gradient vector of a function $f$ on $\mathbb{R}^k$
by $\nabla f$.
That is, $\nabla f$ is the vector
$
(\frac{\partial}{\partial x^1}f ,\ldots, \frac{\partial}{\partial x^k}f )$.
Then, 
we choose a real number $L$ that is larger than the matrix norm of the Hessian of 
$\overline{{\cal G}}_*$, 
which implies the uniform Lipschitz condition;
\begin{align}
\| \nabla \overline{{\cal G}}_*(b)- \nabla \overline{{\cal G}}_*(b')\| 
\le L \| b-b'\| .
\end{align}
Then, we apply the following update rule for the minimization of 
$\overline{{\cal G}}_*(b)$;
\begin{align}
b_{t+1}:= b_t -\frac{1}{L}\nabla \overline{{\cal G}}_*(b_t).\Label{1BVD}
\end{align}
The following precision is guaranteed
\cite[Chapter 10]{Beck} \cite{BT,Nesterov};
\begin{align}
|\overline{{\cal G}}_*(b_k)- \overline{{\cal G}}_*(b_*)| \le 
\frac{L}{2k} \| b_*-b_0\|^2.\Label{XMU}
\end{align}
We notice that
\begin{align}
\nabla \overline{{\cal G}}_*(b)= \argmax_{a \in \mathbb{R}^k} 
\sum_{i=1}^k b^i a_i -\overline{{\cal G}}(a)=
(Q_b[f_i])_{i=1}^k,\Label{ZMR}
\end{align}
where
\begin{align}
Q_b
:=\argmax_{P \in {\cal P}({\cal X})} \sum_{i=1}^k b^i P[f_i]-{\cal G}(P)  
=\argmin_{P \in {\cal P}({\cal X})} 
\sum_{x \in {\cal X}}P(x)
\Big(\Psi[P](x) -\sum_{i=1}^k b^i f_i(x)\Big)
\Label{NER}.
\end{align}
However, the calculation of \eqref{NER} requires a large calculation amount.
Hence, replacing the update rule \eqref{1BVD} 
by a one-step iteration in Algorithm \ref{AL1}, 
we propose another algorithm.

Using 
${\cal F}_{3}^{b}[Q](x):= \frac{1}{\kappa}Q(x)\exp(  -\frac{1}{\gamma}
\Big(\Psi[P](x) -\sum_{i=1}^k b^i f_i(x)\Big)$ with the normalizing constant $\kappa$,
we propose Algorithm \ref{AL4}.

\begin{algorithm}
\caption{Minimization of ${\cal G}(P)$}
\Label{AL4}
\begin{algorithmic}
\STATE {As inputs, we prepare 
the function $\Psi$,
$l$ linearly independent functions $f_1, \ldots, f_{l}$,
a positive number $\gamma>0$, and
the initial values 
$P^{(1)} \in \mathcal{M}_0$, $b_1 \in \mathbb{R}^k$;} 
\REPEAT 
%\STATE {Set $\hat{\theta}=\hat{\theta}^{(t)} \in \Theta_{\mathcal{M}}$;} 
\STATE Calculate $P^{(t+1)}:={\cal F}_{3}^{b_t}[P^{(t)}]$
and $b_{t+1}:= b_t- \frac{1}{L} (P^{(t+1)}[f_i])_{i=1}^k$;
\UNTIL{convergence if it converges. 
If it does not converge, we stop the algorithm at a certain point.
We denote the convergent by $(P^{(\infty)}, b_\infty)$.} 
\STATE{Output $P^{(\infty)}$ and ${\cal G}(P^{(\infty)})$. }
\end{algorithmic}
\end{algorithm}

It is not so easy to evaluate the convergence speed of Algorithm \ref{AL4}.
But, when it converges, the convergent point is the true minimizer.

\begin{theorem}
\if0
Assume that 
the function $\eta \mapsto {\cal G}(P_\eta) $ is $C^2$-continuous and convex.
Also, we assume that 
a real number $L$ is larger than the inverse of the smallest eigenvalue of  
the Hessian of ${\cal G}$.
\fi
When the pair $(b,P)$ is a convergence point,
we have $b=b_*$ and $P=P_*$.
\end{theorem}

\begin{proof}
Since the pair $(b,P)$ is a convergence point,
we have $P={\cal F}_{3}^b[P]$, which implies
\begin{align}
 \sum_{i=1}^k b^i P[f_i]-{\cal G}(P)  
=\sup_{P' \in {\cal P}({\cal X})} \sum_{i=1}^k b^i P'[f_i]-{\cal G}(P')
=\overline{{\cal G}}_*(b). \Label{CMR}
\end{align}
Since the pair $(b,P)$ is a convergence point,
we have $ P[f_i] =0$ for $i=1, \ldots, k$, i.e.,
the distribution $P$ satisfies the required condition in \eqref{MDP}.
the relation \eqref{ZMR} implies
$\nabla \overline{{\cal G}}_*(b)=0$. Hence, \eqref{ZMQ} yields
$\overline{{\cal G}}_*(b)=\overline{{\cal G}}(0)$, which implies $b=b_*$.
Therefore, the relation \eqref{CMR} is rewritten as ${\cal G}(P)  =\overline{{\cal G}}(0)$, which implies $P=P_*$.
\end{proof}

\begin{remark}
We compare our algorithm with 
% extended Arimoto-Blahut Algorithm
a general algorithm proposed in \cite{YSM}.
The input of the objective function in \cite{YSM} forms a mixture family.
The function $f$ given in \cite[(6)]{YSM}
satisfies the condition of ${\cal G}$
by considering the second line of \cite[(6)]{YSM} as $\Psi$.
Their algorithm is the same as Algorithm \ref{AL1} with $\gamma=1$
when there is no constraint
because their extended objective function $g$ defined in \cite[(16)]{YSM}
can be considered as $D(P\|Q)+ \sum_{x \in {\cal X}}P(x) \Psi[Q](x)$,
where the choice of $q$ in \cite{YSM} corresponds to the choice of $P$
and the choice of $Q_1,\ldots, Q_K$ in \cite{YSM} does to the choice of $Q$.

Also, we can show that the function $f$ given in \cite[(6)]{YSM} satisfies the condition (A4).
Since the condition (A4) holds, the convexity of $f$ is equivalent to
the condition (B1). This equivalence, in this case, was shown as \cite[Proposition 4.1]{YSM}.
They showed the convergence of their algorithm as \cite[Theorem 4.1]{YSM},
which can be considered as a special case of our Theorem \ref{TH1}.

However, our treatment for the constraint is different from theirs.
They consider the minimization
$\min_{P \in {\cal P}({\cal X})} {\cal G}(P)-\sum_{i=1}^k b^i P[f_i] $
without updating the parameter $b$. Hence, their algorithm cannot achieve
the minimum with the desired constraint
while Algorithms \ref{AL1}, \ref{AL3}, and \ref{AL4}
achieve the minimum with the desired constraint.
Although their algorithm is similar to Algorithm \ref{AL4},
Algorithm \ref{AL4} updates the parameter $b$ to 
achieve the minimum with the desired constraint.
\end{remark}

\section{Application to information theoretical problems}\Label{S4}
\subsection{Channel capacity}\Label{S4-1}
In the same way as the reference \cite{RISB}, we apply our problem setting to the channel coding.
A channel is given as a conditional distributions $W_{Y|X=x}$ on 
the sample space ${\cal Y}$ with conditions on the sample space ${\cal X}$, 
where ${\cal Y}$ is a general sample space with a measure $\mu$ and ${\cal X}$
is a finite sample space.
For two absolutely continuous distributions $P_Y$ and $Q_Y$ with respect to $\mu$
on ${\cal Y}$,
the Kullback-Leibler divergence $D(P_Y\|Q_Y)$ is given as
\begin{align}
D(P_Y\|Q_Y):= \int_{{\cal Y}}p_Y(y) (\log p_Y(y)-\log q_Y(y))\mu(dy),
\end{align}
where $p_Y$ and $q_Y$ are the probability density functions of $P_Y$ and $Q_Y$ with respect to $\mu$.
This quantity is generalized to the Renyi divergence with order $\alpha>0$ 
as
\begin{align}
D_\alpha(P_{Y} \| Q_Y):=
\frac{1}{\alpha-1}\log \int_{{\cal Y}} \Big(\frac{p_Y(y)}{q_Y(y)}\Big)^{\alpha-1}p_Y(y) \mu(dy).
\end{align}
The channel capacity $C(W_{Y|X})$ is given as the maximization of the mutual information $I(P_X,W_{Y|X})$ as \cite{Shannon}
\begin{align}
C(W_{Y|X})&:=\max_{P_X} I(P_X,W_{Y|X}) \Label{CMD}\\
I(P_X,W_{Y|X})&:=\sum_{x \in {\cal X}}P_X(x) D(W_{Y|X=x}\| W_{Y|X} \cdot P_X) \nonumber \\
&= D(W_{Y|X} \times P_X \| (W_{Y|X} \cdot P_X) \times P_X),
\end{align}
where $W_{Y|X} \cdot P_X$ and $W_{Y|X} \times P_X$ are defined as the following probability density functions
$w_{Y|X} \cdot P_X$ and $w_{Y|X} \times P_X$;
\begin{align}
(w_{Y|X} \cdot P_X)(y) &:= \sum_{x \in {\cal X}}P_X(x) w_{Y|X=x}(y) \\
(w_{Y|X} \times P_X)(x, y)& := P_X(x) w_{Y|X=x}(y) .
\end{align}
However, the mutual information $I(P_X,W_{Y|X})$ has another form as %the following minimax problem.
\begin{align}
I(P_X, W_{Y|X})= \min_{Q_Y} \sum_{x \in {\cal X}}P_X(x) D(W_{Y|X=x}\| Q_{Y}).
\Label{Eq82}
\end{align}
%\subsubsection{1st type of application}
When we choose ${\cal M}_a$ and  $\Psi$ as ${\cal P}({\cal X})$ and 
\begin{align}
\Psi_{W_{Y|X}}[P_X](x):= - D(W_{Y|X=x}\| W_{Y|X} \cdot P_X),
\end{align}
$-I(P_X,W_{Y|X})$ coincides with ${\cal G}(P_X)$ \cite{RISB}.
Since
\begin{align}
D_{\Psi}(P_X\| Q_X)
%\sum_{x\in {\cal X}}P_X(x) (\Psi[P_X](x)-\Psi[Q_X](x))
=
D(W_{Y|X} \cdot P_X\| W_{Y|X} \cdot Q_X) \ge 0,
\end{align}
the condition (A2) holds with ${\cal P}({\cal X})$.
In addition, since the information processing inequality guarantees that
\begin{align}
D(W_{Y|X} \cdot P_X\| W_{Y|X} \cdot Q_X) \le D(P_X\|Q_X),
\end{align}
the condition (A1) holds with $\gamma=1$ and ${\cal P}({\cal X})$.
In this case, ${\cal F}_3$ is given as 
\begin{align}
{\cal F}_3[Q_X](x)=\frac{1}{\kappa_{W_{Y|X}}[Q_X]}Q_X(x) \exp (  \frac{1}{\gamma}  
D(W_{Y|X=x}\| W_{Y|X} \cdot Q_X)),
\end{align}
where the normalizing constant $\kappa_{W_{Y|X}}[Q_X]$ is given as
\begin{align}
\kappa_{W_{Y|X}}[Q_X]=\sum_{x \in {\cal X}}Q_X(x) 
\exp (  \frac{1}{\gamma}  D(W_{Y|X=x}\| W_{Y|X} \cdot Q_X)).
\end{align}
When $\gamma=1$, it coincides with the Arimoto-Blahut algorithm \cite{Arimoto,Blahut}. 
Since ${\cal F}_3[Q_X] \in {\cal P}({\cal X})$, 
$P_X^{(t+1)}$ is given as ${\cal F}_3[P_X^{(t)}]$.
%Calculate $P^{(t+1)}:=\Gamma^{(e)}_{{\cal M}_a}[{\cal F}_3[P^{(t)}]]$;

\begin{remark}
The reference \cite{RISB} covers the case when
${\cal M}_a$ is given as ${\cal P}({\cal X})$, and
the reference \cite{RISB} presented
the algorithms presented in this subsection 
in a more general form.
%the case when $Y$ is a quantum system.
Also, they proposed an adaptive choice of $\gamma$ in this case \cite[(22)]{RISB}.
In addition, they numerically compared 
their adaptive choice with the case of $\gamma=1$ \cite[Figs. 1,..., 6]{RISB}.
These comparisons show a significant improvement by their adaptive choice.
\end{remark}

\subsection{Reliability function in channel coding}\Label{S4-2}
In channel coding, we consider the reliability function, which was originally introduced by 
Gallager \cite{Gallager} and expresses the exponential decreasing rate of 
an upper bound of the decoding block error probability under the random coding.
To achieve this aim, for $\alpha > 0 $, we define
\begin{align}
I_{\alpha}(P_X,W_{Y|X}):= \frac{\alpha}{\alpha-1}\log \Big(
\int_{{\cal Y}} \Big(\sum_{x \in {\cal X}} P_X(x) w_{Y|X=x}(y)^{\alpha}\Big)^{\frac{1}{\alpha}} \mu(dy)
\Big).
\end{align}
Then, when the code is generated with the random coding based on the distribution $P_X$,
the decoding block error probability with coding rate $R$ is upper bounded by the following quantity;
\begin{align}
e^{n \min_{\rho \in [0,1]}
\big(\rho R -\rho I_{\frac{1}{1+\rho}}(P_X,W_{Y|X})\big)}
\end{align}
when we use the channel $W_{Y|X}$ with $n$ times.
Notice that $e^{-\rho I_{\frac{1}{1+\rho}}(P_X,W_{Y|X})}=
\int_{{\cal Y}} \Big(\sum_{x \in {\cal X}} P_X(x) w_{Y|X=x}(y)^{\frac{1}{1+\rho}}\Big)^{1+\rho} \mu(dy)$.
That is, the Gallager function \cite{Gallager} is given as
$\rho I_{\frac{1}{1+\rho}}(P_X,W_{Y|X})$, i.e., 
the parameter $\alpha$ is different from 
the parameter $\rho$ in the Gallager function.
Taking the minimum for the choice of $P_X$, we have
\begin{align}
\min_{P_X} e^{\min_{\rho \in [0,1]}\rho R -\rho I_{\frac{1}{1+\rho}}(P_X,W_{Y|X})}
=
%\min_{\rho \in [0,1]} \Big(e^{\rho R} \min_{P_X} e^{-\rho I_{\frac{1}{1-\rho}}(P_X)}\Big)
\min_{\alpha \in [1/2,1]} \Big(e^{-\frac{\alpha-1}{\alpha} R} \min_{P_X} 
e^{\frac{\alpha-1}{\alpha} I_{\alpha}(P_X,W_{Y|X})}\Big)
\end{align}
with $\alpha=\frac{1}{1-\rho}\in [1/2,1]$.
Therefore, we consider the following minimization;
\begin{align}
\min_{P_X} e^{\frac{\alpha-1}{\alpha} I_{\alpha}(P_X,W_{Y|X})}
=\min_{P_X}
\int_{{\cal Y}} \Big(\sum_{x \in {\cal X}} P_X(x) w_{Y|X=x}(y)^{\alpha}\Big)^{\frac{1}{\alpha}} \mu(dy)
\Label{AMP}.
\end{align} 
In the following, we discuss the RHS of \eqref{AMP} with $\alpha \in [1/2,1]$.

To apply our method, 
as a generalization of \eqref{Eq82},
we consider another expression of $I_{\alpha}(P_X,W_{Y|X})$;
\begin{align}
I_{\alpha}(P_X,W_{Y|X})= \min_{Q_Y} D_\alpha(W_{Y|X}\times P_X \| Q_Y\times P_X),
\Label{MXP}
\end{align}
which was shown in \cite[Lemma 2]{H15}.
Using 
\begin{align}
Q_{Y|\alpha,P_X}:=& \argmin_{Q_Y} D_\alpha(W_{Y|X}\times P_X \| Q_Y\times P_X) \nonumber\\
=&\argmax_{Q_Y}\sum_{x\in {\cal X}}P_X(x)e^{ (\alpha-1)D_\alpha(W_{Y|X=x} \| Q_{Y})},\Label{XCP}
\end{align}
we have
\begin{align}
\Big(\min_{P_X} e^{\frac{\alpha-1}{\alpha} I_{\alpha}(P_X,W_{Y|X})}\Big)^\alpha
=
 \min_{P_X}
\sum_{x\in {\cal X}}P_X(x)e^{ (\alpha-1) D_\alpha(W_{Y|X=x} \| Q_{Y|\alpha,P_X})}.
\Label{APT}
\end{align} 
The probability density function $q_{Y|\alpha,P_X}$ of
the minimizer $Q_{Y|\alpha,P_X}$ is calculated as
\begin{align}
q_{Y|\alpha,P_X}(y)=
C 
\Big(\sum_{x \in {\cal X}} P_X(x) w_{Y|X=x}(y)^{\alpha}\Big)^{\frac{1}{\alpha}} ,
\end{align}
where $C$ is the normalized constant \cite[Lemma 2]{H15}.

To solve the minimization \eqref{APT}, we apply our method to the case when 
we choose ${\cal M}_a$ and  $\Psi$ as ${\cal P}({\cal X})$ and 
\begin{align}
\Psi_{\alpha,W_{Y|X}}[P_X](x):= e^{(\alpha-1) D_\alpha(W_{Y|X=x} \| Q_{Y|\alpha,P_X})}.
\end{align}
Since \eqref{XCP} guarantees that
\begin{align}
&\sum_{x\in {\cal X}}P_X(x) (\Psi_{\alpha,W_{Y|X}}[P_X](x)
-\Psi_{\alpha,W_{Y|X}}[Q_X](x)) \nonumber \\
=&
\sum_{x\in {\cal X}}P_X(x)
\Big(e^{ (\alpha-1)D_\alpha(W_{Y|X=x} \| Q_{Y|\alpha,P_X})}
-
e^{ (\alpha-1)D_\alpha(W_{Y|X=x} \| Q_{Y|\alpha,Q_X})}\Big)
 \ge 0,
\end{align}
the condition (A2) holds with ${\cal P}({\cal X})$.
The condition (A1) can be satisfied with sufficiently large $\gamma$.
In this case, ${\cal F}_3$ is given as 
\begin{align}
{\cal F}_{3,\alpha}[Q_X](x)=\frac{1}{\kappa_{\alpha,W_{Y|X}}[Q_X]}Q_X(x) \exp (  -\frac{1}{\gamma}  
e^{(\alpha-1) D_\alpha(W_{Y|X=x} \| Q_{Y|\alpha,P_X})}),
\end{align}
where the normalizing constant $\kappa_{\alpha,W_{Y|X}}[Q_X]$ is given as
$\kappa_{\alpha,W_{Y|X}}[Q_X]=$\par\noindent
$\sum_{x \in {\cal X}}Q_X(x)  \exp (  -\frac{1}{\gamma}  
e^{(\alpha-1) D_\alpha( W_{Y|X=x} \| Q_{Y|\alpha,P_X})})$.
Since ${\cal F}_{3,\alpha}[Q_X] \in {\cal M}_a$, 
$P_X^{(t+1)}$ is given as ${\cal F}_{3,\alpha}[P_X^{(t)}]$.

\subsection{Strong converse exponent in channel coding}\Label{S4-3}
In channel coding, we discuss an upper bound of the probability of correct decoding.
This probability is upper bounded by  
the following quantity;
\begin{align}
%\max_{P_X} e^{n \min_{\rho \in [-1,0]}\rho R -\rho I_{\frac{1}{1+\rho}}(P_X)}
\max_{P_X} e^{n \min_{\rho \in [0,1]} 
\Big(-\rho R +\rho I_{\frac{1}{1-\rho}}(P_X,W_{Y|X})\Big)}
\end{align}
when we use the channel $P_{Y|X}$ with $n$ times and the coding rate is $R$ \cite{Arimoto2}.
Therefore, we consider the following maximization;
\begin{align}
\max_{P_X} e^{\rho I_{\frac{1}{1-\rho}}(P_X,W_{Y|X})}
=\max_{P_X} e^{\frac{\alpha-1}{\alpha} I_{\alpha}(P_X,W_{Y|X})}
%\min_{P_X} e^{\frac{\alpha-1}{\alpha} I_{\alpha}(P_X)}
%=\min_{P_X}
%\int_{{\cal Y}} \Big(\sum_{x \in {\cal X}} P_X(x) p_{Y|X=x}(y)^{\alpha}\Big)^{\frac{1}{\alpha}} \mu(dy)
\Label{AMP2}
\end{align} 
with $\alpha=\frac{1}{1-\rho}>1$.
In the following, we discuss the RHS of \eqref{AMP2} with $\alpha >1$.

To apply our method, we consider another expression \eqref{MXP} of 
$I_{\alpha}(P_X,W_{Y|X})$.
Using \eqref{XCP}, we have
\begin{align}
\Big(\max_{P_X} e^{\frac{\alpha-1}{\alpha} I_{\alpha}(P_X,W_{Y|X})}\Big)^\alpha
=
 \max_{P_X}
\sum_{x\in {\cal X}}P_X(x)e^{ (\alpha-1) D_\alpha(W_{Y|X=x} \| Q_{Y|\alpha,P_X})}.
\Label{APT2}
\end{align} 
The maximization \eqref{APT2} can be solved by choosing 
${\cal M}_a$ and  $\Psi$ as ${\cal P}({\cal X})$ and 
\begin{align}
\Psi_{\alpha,W_{Y|X}}[P_X](x):= - e^{(\alpha-1) D_\alpha(W_{Y|X=x} \| Q_{Y|\alpha,P_X})}.
\end{align}
Since \eqref{XCP} guarantees that
\begin{align}
&\sum_{x\in {\cal X}}P_X(x) (\Psi_{\alpha,W_{Y|X}}[P_X](x)
-\Psi_{\alpha,W_{Y|X}}[Q_X](x)) \nonumber \\
=&
\sum_{x\in {\cal X}}P_X(x)
\Big(-e^{ (\alpha-1)D_\alpha(W_{Y|X=x} \| Q_{Y|\alpha,P_X})}
+
e^{ (\alpha-1)D_\alpha(W_{Y|X=x} \| Q_{Y|\alpha,Q_X})}\Big)
 \ge 0,
\end{align}
the condition (A2) holds with ${\cal P}({\cal X})$.
Similarly, the condition (A1) can be satisfied with sufficiently large $\gamma$.
In this case, ${\cal F}_3$ is given as 
\begin{align}
{\cal F}_{3,\alpha}[Q_X](x)=\frac{1}{\kappa_{\alpha,W_{Y|X}}[Q_X]}Q_X(x) \exp ( \frac{1}{\gamma}  
e^{(\alpha-1) D_\alpha( W_{Y|X=x} \| Q_{Y|\alpha,P_X})}),
\end{align}
where the normalizing constant $\kappa_{\alpha,W_{Y|X}}[Q_X]$ is given as
$\kappa_{\alpha,W_{Y|X}}[Q_X]=$\par\noindent
$\sum_{x \in {\cal X}}Q_X(x)  \exp ( \frac{1}{\gamma}  
e^{(\alpha-1) D_\alpha(W_{Y|X=x} \| Q_{Y|\alpha,P_X})})$.
Since ${\cal F}_{3,\alpha}[Q_X] \in {\cal M}_a$, 
$P_X^{(t+1)}$ is given as ${\cal F}_{3,\alpha}[P_X^{(t)}]$.

\subsection{Wiretap channel capacity}
\subsubsection{General case}
Given a pair of a channel $W_{Y|X}$ from ${\cal X}$ to a legitimate user ${\cal Y}$
and 
a channel $W_{Z|X}$ from ${\cal X}$ to a malicious user ${\cal Z}$,
the wiretap channel capacity is given as \cite{Wyner,CK79}
\begin{align}
C(W_{Y|X},W_{Z|X}):=\max_{P_{VX}} I(P_V, W_{Y|X}\cdot P_{X|V})-I(P_V, W_{Z|X}\cdot P_{X|V})\Label{XAT}
\end{align}
with a sufficiently large discrete set ${\cal V}$.
The recent papers showed that the above rate can be achieved even with 
strong security \cite{Csisz,H06,H11} and the semantic security \cite{BTV,HM16}.
Furthermore, the paper \cite[Appendix D]{HM16} showed the above even when 
the output systems are general continuous systems including Gaussian channels.
The wiretap capacity \eqref{XAT} can be calculated via the minimization;
\begin{align}
\min_{P_{VX}} -I(P_V, W_{Y|X}\cdot P_{X|V})+I(P_V, W_{Z|X}\cdot P_{X|V}).\Label{XAT2}
\end{align}
Here, ${\cal V}$ is an additional discrete sample space. 
When we choose ${\cal M}_a$ and  $\Psi$ as ${\cal P}({\cal X} \times {\cal V})$ and 
\begin{align}
&\Psi_{W_{Y|X},W_{Z|X}}[P_{VX}](v,x)\nonumber \\
:= 
&D(W_{Z|X}\cdot P_{X|V=v}\|W_{Z|X}\cdot P_{X} ) 
-D(W_{Y|X}\cdot P_{X|V=v}\|W_{Y|X}\cdot P_{X} ),
\end{align}
$-I(P_V, W_{Y|X}\cdot P_{X|V})+I(P_V, W_{Z|X}\cdot P_{X|V})$ coincides with ${\cal G}(P_{VX})$.
Hence, the general theory in Section \ref{setup} can be used for the minimization \eqref{XAT2}.
\if0
\begin{align}
&D_{\Psi_{W_{Y|X},W_{Z|X}}[P_{VX}]}(P_{VX}\|Q_{VX}) \\
=&
\sum_{x,v} P_{VX}(v,x)
\Big(D(W_{Z|X}\cdot P_{X|V=v}\|W_{Z|X}\cdot P_{X} ) \\
&-D(W_{Y|X}\cdot P_{X|V=v}\|W_{Y|X}\cdot P_{X} )
-D(W_{Z|X}\cdot Q_{X|V=v}\|W_{Z|X}\cdot Q_{X} ) \\
&+D(W_{Y|X}\cdot Q_{X|V=v}\|W_{Y|X}\cdot Q_{X} )\Big)\\
\le &
\sum_{v} P_{V}(v)
\Big(D(W_{Z|X}\cdot P_{X|V=v}\|W_{Z|X}\cdot P_{X} ) \\
&+D(W_{Y|X}\cdot Q_{X|V=v}\|W_{Y|X}\cdot Q_{X} )\Big)\\
\le &
\sum_{v} P_{V}(v)
\Big(D(P_{X|V=v}\| P_{X} ) +D(Q_{X|V=v}\|Q_{X} )\Big)\\
\end{align}

$D_{\Psi}(P \|Q)=\sum_{x \in {\cal X}} P(x) (\Psi[P](x)-\Psi[Q](x)) \ge 0$
\fi
In this case, it is difficult to clarify whether the conditions (A1) and (A2) hold in general.
${\cal F}_3$ is given as 
\begin{align}
&{\cal F}_3[Q_{VX}](v,x)\nonumber\\
=&\frac{1}{\kappa_{W_{Y|X},W_{Z|X}}[Q_{VX}]}Q_{VX}(v,x) 
\exp \Big(  \frac{1}{\gamma}  \Big( D(W_{Y|X}\cdot P_{X|V=v}\|W_{Y|X}\cdot P_{X} )\nonumber\\
&-D(W_{Z|X}\cdot P_{X|V=v}\|W_{Z|X}\cdot P_{X} ) \Big) \Big),
\end{align}
where $\kappa_{W_{Y|X},W_{Z|X}}[Q_{VX}]$ is the normalizing constant.
Since ${\cal F}_3[Q_X] \in {\cal M}_a$, 
$P_X^{(t+1)}$ is given as ${\cal F}_3[P_X^{(t)}]$.

\subsubsection{Degraded case}\Label{S4-5-2}
However, when there exists a channel $W_{Z|Y}$ from ${\cal Y}$ to ${\cal Z}$ such that
$ W_{Z|Y} \cdot W_{Y|X}=W_{Z|X}$, i.e., the channel $W_{Z|X}$ is a degraded channel of $W_{Y|X}$,
we can define the joint channel $W_{YZ|X}$ with the following conditional probability density function 
\begin{align}
w_{YZ|X}(yz|x):= w_{Z|Y}(z|y)w_{Y|X}(y|x).
\end{align}
Then, the maximization \eqref{XAT} is simplified as
\begin{align}
C(W_{YZ|X}):=\max_{P_{X}} I(X;Y|Z)[P_X, W_{YZ|X}] \Label{ZLO}
\end{align} 
where the conditional mutual information is given as 
\begin{align}
I(X;Y|Z)[P_X, W_{YZ|X}] := 
\sum_{x,z} P_{XZ}(x,z) D(P_{Y|X=x,Z=z}\| P_{Y|Z=z}),
\Label{XATE}
\end{align} 
where the conditional distributions $P_{Y|XZ}$ and $P_{Y|Z}$
are defined from the joint distribution $W_{YZ|X}\times P_X$.
To consider \eqref{ZLO}, we consider
the following minimization
with a general two-output channel $W_{YZ|X}$;
\begin{align}
\min_{P_{X}} - I(X;Y|Z)[P_X, W_{YZ|X}] \Label{ZLO2}.
\end{align} 
When we choose ${\cal M}_a$ and  $\Psi$ as ${\cal P}({\cal X})$ and 
\begin{align}
\Psi_{W_{YZ|X}}[P_{X}](x):= 
- \sum_{z} P_{Z|X=x}(z) D(P_{Y|X=x,Z=z}\| P_{Y|Z=z}).
\end{align}
$- I(X;Y|Z)[P_X, W_{YZ|X}]$ coincides with ${\cal G}(P_{X})$.
Hence, the general theory in Section \ref{setup} can be used for the minimization \eqref{XAT2}.
In this case, as shown in Subsection \ref{MXT}, the conditions (A1) with $\gamma=1$ and (A2) hold.
${\cal F}_3$ is given as 
\begin{align}
&{\cal F}_3[Q_{X}](x)\nonumber\\
=&\frac{1}{\kappa_{W_{YZ|X}}[Q_{X}]}Q_{X}(x) 
\exp \Big(  \frac{1}{\gamma}  \Big( \sum_{z} P_{Z|X=x}(z) D(P_{Y|X=x,Z=z}\| P_{Y|Z=z}) \Big) \Big),
\end{align}
where $\kappa_{W_{YZ|X}}[Q_{X}]$ is the normalizing constant.
Since ${\cal F}_3[Q_X] \in {\cal M}_a$, 
$P_X^{(t+1)}$ is given as ${\cal F}_3[P_X^{(t)}]$.
The above algorithm with $\gamma=1$ coincides with 
the algorithm proposed by \cite{Yasui}.

\subsection{Capacities with cost constraint}\Label{S4-4}
Next, we consider the case when a cost constraint is imposed.
Consider a function $f$ on ${\cal X}$ and the following constraint for a distribution 
$P_X \in {\cal X}$;
\begin{align}
P_X[f]=a.\Label{COS}
%\sum_{x \in {\cal X}} 
\end{align}
We define ${\cal M}_a$ by imposing the condition \eqref{COS}
as a special case of \eqref{MDP}.
The capacity of the channel $W_{Y |X}$ 
under the cost constraint is given as
$\max_{P_X \in {\cal M}_a}I(P_X,W_{Y|X})$.
That is, we need to solve the minimization 
$\min_{P_X \in {\cal M}_a}-I(P_X,W_{Y|X})$.
In this case, the $t+1$-th distribution $P^{(t+1)}$ is given as
$\Gamma_{{\cal M}_a}^{(e)} [{\cal F}_3[P_X^{(t)}]]$.
Since $\Gamma_{{\cal M}_a}^{(e)} [{\cal F}_3[P_X^{(t)}]]$ cannot be calculated analytically,
we can use Algorithm \ref{AL3} instead of Algorithm \ref{AL1}.
Since conditions (A1) with $\gamma=1$ and (A2)
hold, Theorem \ref{TH8} guarantees the global convergence to the minimum in Algorithm \ref{AL3}.

We can consider the cost constraint \eqref{COS} for 
the problems \eqref{APT} and \eqref{APT2}.
In these cases, we have a similar modification by considering 
$\Gamma_{{\cal M}_a}^{(e)} [{\cal F}_3[P_X^{(t)}]]$.

\section{em problem}\Label{S5}
%\subsection{}
We apply our algorithm to the problem setting with the em algorithm \cite{Amari,Fujimoto,Allassonniere}, which is 
a generalization of Boltzmann machines \cite{Bol}.
The em algorithm is implemented by 
iterative applications of
the projection to an exponential family (the $m$-projection)
and the projection to a mixture family (the $e$-projection).
Hence, this algorithm is called the em algorithm.
In contrast,
EM algorithm is implemented by 
iterative applications of expectation and maximization.
Their relation is summarized as follows.
In particular, the expectation in the EM algorithm, which is often 
called E-step, corresponds to 
the $e$-projection to a mixture family, which is often called
$e$-step in the em algorithm.
Also, the maximization in the EM algorithm, which is often 
called M-step, corresponds to 
the $m$-projection to an exponential family, which is often called
$m$-step in the em algorithm.
In this reason, they are essentially the same \cite{Amari}.

For this aim, we consider a pair of an exponential family ${\cal E}$ and 
a mixture family ${\cal M}_a$ on ${\cal X}$.
We denote the $m$-projection to ${\cal E}$ of $P$ 
by $\Gamma_{\cal E}^{(m)}[P]$, which is defined as \cite{Amari1,Amari}
\begin{align}
\Gamma_{\cal E}^{(m)}[P]:=
\argmin_{Q\in {\cal E}} D(P\|Q).
\end{align}
We consider the following minimization;
\begin{align}
&\min_{P\in {\cal M}_a}\min_{Q\in {\cal E}} D(P\|Q)=
\min_{P\in {\cal M}_a} D(P\|\Gamma_{\cal E}^{(m)}[P]) \nonumber \\
=&
\min_{P\in {\cal M}_a} \sum_{x \in {\cal X}}P(x)
(\log P(x) - \log \Gamma_{\cal E}^{(m)}[P](x)).\Label{AMR}
\end{align}
We choose the function $\Psi$ as
\begin{align}
\Psi_{\rm{em}}[P](x):= (\log P(x) - \log \Gamma_{\cal E}^{(m)}[P](x)),\Label{ZSP}
\end{align}
and apply the discussion in Section \ref{setup}.
Then, 
we have
\begin{align}
&\sum_{x \in {\cal X}} P^{0}(x) (\Psi_{\rm{em}}[P^{0}](x)- \Psi_{\rm{em}}[Q](x)) \nonumber \\
=&\sum_{x \in {\cal X}} P^{0}(x) \Big(
(\log P^{0}(x) - \log \Gamma_{\cal E}^{(m)}[P^{0}](x))
-
(\log Q(x) - \log \Gamma_{\cal E}^{(m)}[Q](x)) \Big) \nonumber \\
=& 
\sum_{x \in {\cal X}} P^{0}(x)(\log P^{0}(x)-\log Q(x)) \nonumber \\
&+
\sum_{x \in {\cal X}} P^{0}(x)
\Big( \log \Gamma_{\cal E}^{(m)}[Q](x))
 - \log \Gamma_{\cal E}^{(m)}[P^{0}](x)) \Big) \nonumber \\
=& 
D (P^{0}\| Q) 
+
D(P^{0} \| \Gamma_{\cal E}^{(m)}[P^{0}])
-D(P^{0} \| \Gamma_{\cal E}^{(m)}[Q])
\nonumber \\
=& 
D (P^{0}\| Q) 
-D(\Gamma_{\cal E}^{(m)}[P^{0}] \| 
\Gamma_{\cal E}^{(m)}[Q]),\Label{ZME}
\end{align}
where the final equation follows from \eqref{NNP}.
The condition (A1) holds with $U(P^0,\infty)={\cal M}_a$ and $\gamma=1$.
There is a possibility that the condition (A1) holds with a smaller $\gamma$.
Therefore, with $\gamma=1$,
Theorem \ref{TTH1} guarantees that 
Algorithm \ref{AL1} converges to a local minimum.
In addition, 
when the relation 
\begin{align}
D (P^0\| Q) 
\ge D(\Gamma_{\cal E}^{(m)}[P^0] \| 
\Gamma_{\cal E}^{(m)}[Q]) \Label{AMO}
\end{align}
holds for $Q \in U(P^0,\delta)$, the condition (A2) holds with $U(P^0,\delta)$.
That is, if the condition \eqref{AMO} holds, Algorithm \ref{AL1} has the global convergence to the minimizer. 
The condition \eqref{AMO} is a similar condition to the condition given in \cite{Bregman-em}.

In this case, ${\cal F}_3$ is given as 
\begin{align}
{\cal F}_3[Q](x)
=&\frac{1}{\kappa_{\rm{em}}[Q]}Q(x) \exp \Big( - \frac{1}{\gamma}
(\log Q(x) - \log \Gamma_{\cal E}^{(m)}[Q](x))\Big)\nonumber \\
=&\frac{1}{\kappa_{\rm{em}}[Q]}Q(x)^{\frac{\gamma-1}{\gamma}} 
\Gamma_{\cal E}^{(m)}[Q](x)^{\frac{1}{\gamma}},\Label{AMD}
\end{align}
where the normalizing constant $\kappa_{\rm{em}}[Q_X]$ is given as
$\kappa_{\rm{em}}[Q_X]=\sum_{x \in {\cal X}}
Q(x)^{\frac{\gamma-1}{\gamma}} 
\Gamma_{\cal E}^{(m)}[Q](x)^{\frac{1}{\gamma}}$.
Since ${\cal F}_3[Q] \in {\cal M}_a$, 
$P^{(t+1)}$ is given as $\Gamma_{{\cal M}_a}^{(e)}[{\cal F}_3[P^{(t)}]]$.
When $\gamma=1$, it coincides with the conventional em-algorithm \cite{Amari,Fujimoto,Allassonniere}
because ${\cal F}_3[P^{(t)}]=\Gamma_{\cal E}^{(m)}[P^{(t)}]$.
The above analysis suggests 
the choice of $\gamma$ as a smaller value than $1$.
That is, there is a possibility that a smaller $\gamma$ improves the conventional em-algorithm. 
In addition, we may use Algorithm \ref{AL3} instead of Algorithm \ref{AL1} when
the calculation of $e$-projection is difficult.

\begin{lemma}\Label{L9}
When $\Psi_{\rm{em}}$ is given as \eqref{ZSP}, 
the condition (A4) holds.
\end{lemma}

Therefore, 
by combining Lemmas \ref{L6} and \ref{L9},
the assumption of Theorem \ref{TH1} holds 
in the $\delta$ neighborhood of a local minimizer
with sufficiently small $\delta>0$.
That is, the convergence speed can be evaluated by 
Theorem \ref{TH1}.

\begin{proof}
Pythagorean theorem guarantees 
\begin{align}
&\sum_{x \in {\cal X}}P (x) \big(
\log \Gamma_{\cal E}^{(m)}[Q](x)-
\log \Gamma_{\cal E}^{(m)}[P](x)
\big) \nonumber \\
=& D(P\| \Gamma_{\cal E}^{(m)}[P]) 
- D(P\| \Gamma_{\cal E}^{(m)}[Q]) 
=  D(\Gamma_{\cal E}^{(m)}[P]\| 
\Gamma_{\cal E}^{(m)}[Q]) \Label{ZAE}.
\end{align}
We make the parameterization $P_\eta \in {\cal M}_a$ with mixture parameter
$\eta$.
We denote $\eta(h,i):=(\eta(0)_1, \ldots,  \eta(0)_{i-1},\eta(0)_i+h,\eta(0)_{i+1},\ldots,
\eta(0)_k)$.
\begin{align}
&\sum_{x \in {\cal X}}P_{\eta(0)} (x) 
\Big(\frac{\partial }{\partial \eta_i}\Psi[P_\eta](x)|_{\eta=\eta(0)}\Big)
\nonumber\\
=&
\sum_{x \in {\cal X}}P_{\eta(0)} (x) 
\Big( 
\lim_{h\to 0}\frac{\Psi[P_{\eta(h,i)}](x)-\Psi[P_{\eta(0)}](x)}{h}
\Big) \nonumber\\
=&
\sum_{x \in {\cal X}}P_{\eta(0)} (x) 
\Big( 
\lim_{h\to 0}\frac{ \log P_{\eta(h,i)}(x)-\log P_{\eta(0)}(x)}{h} \nonumber\\
&-\lim_{h\to 0}\frac{ \log \Gamma_{\cal E}^{(m)}[P_{\eta(h,i)}](x)-
\log \Gamma_{\cal E}^{(m)}[P_{\eta(0)}](x)}{h}
\Big) \nonumber\\
\stackrel{(a)}{=} & 
\sum_{x \in {\cal X}}P_{\eta(0)} (x) 
\Big( 
\lim_{h\to 0}\frac{ \log P_{\eta(h,i)}(x)-\log P_{\eta(0)}(x)}{h} \Big)\nonumber\\
&-
\sum_{x \in {\cal X}}\Gamma_{\cal E}^{(m)}[P_{\eta(0)}] (x) 
\Big( 
\lim_{h\to 0}
\frac{ \log \Gamma_{\cal E}^{(m)}[P_{\eta(h,i)}](x)-
\log \Gamma_{\cal E}^{(m)}[P_{\eta(0)}](x)}{h}
\Big) \nonumber\\
=&
\sum_{x \in {\cal X}}
\frac{\partial }{\partial \eta_i} P_\eta(x)|_{\eta=\eta(0)}
-
\sum_{x \in {\cal X}}
\frac{\partial }{\partial \eta_i} \Gamma_{\cal E}^{(m)}[P_\eta](x)|_{\eta=\eta(0)}
=0,
\end{align}
which implies the condition (A4).
Here, $(a)$ follows from \eqref{ZAE}.
\end{proof}

\section{Commitment capacity}\Label{S6}
Using the same notations as Section \ref{S4}, we address
the bit commitment via a noisy channel $W_{Y|X}$.
Given a distribution $P_X$, the Shannon entropy is given as
\begin{align}
H(X)_{P_X}:=-\sum_{x \in {\cal X}}P_X(x)\log P_X(x).
\end{align}
Given a joint distribution $P_{XY}$, the conditional entropy
is defined as
\begin{align}
H(X|Y)_{P_{XY}}:=
\int_{{\cal Y}} H(X)_{W_{X|Y=y}} p_Y(y) \mu(dy).
\end{align}
The commitment capacity is given as
\begin{align}
C_c(W_{Y|X}):=\max_{P_X} H(X|Y)_{W_{Y|X} \times P_X}.\Label{BCO6}
\end{align}

This problem setting has several versions.
To achieve the bit commitment,
the papers \cite{BC1,CCDM,W-Protocols} considered interactive protocols with multiple rounds,
where each round has one use of the given noisy channel $W_{Y|X}$ and 
free noiseless communications in both directions. 
Then, it derived the commitment capacity \eqref{BCO6}.
Basically, the proof is composed of two parts.
One is the achievability part, which is often called the direct part and
shows the existence of the code to achieve the capacity.
The other is the impossibility part, which is often called the converse part and
shows the non-existence of the code to exceed the capacity.
As the achievability part,
they showed that the commitment capacity can be achieved with 
non-interactive protocol, which has no free noiseless communication 
during multiple uses of the given noisy channel $W_{Y|X}$.
However, as explained in \cite{HW22},
their proof of the impossibility part skips so many steps that it cannot be followed.
Later, the paper \cite{BC3} showed the impossibility part only for non-interactive protocols
by applying the wiretap channel.
Recently, the paper \cite{H21} constructed a code to achieve the commitment capacity 
by using a specific type of list decoding. 
Further, the paper showed that the achievability 
of the commitment capacity even in the quantum setting.
In addition, the paper \cite{HW22} showed the impossibility part for interactive protocols
by completing the proof by \cite{HW22}. 
The proof in \cite{HW22} covers the impossibility part for a certain class even in the quantum setting.

\subsection{Algorithm based on em-algorithm problem}
To calculate the commitment capacity, 
we consider the following mixture and exponential families;
\begin{align}
{\cal M}_a &:=\{W_{Y|X} \times P_X| P_X \in {\cal P}({\cal X})\} \Label{CPA}\\
{\cal E} &:=\{ Q_Y \times P_{X,Uni} | Q_Y \in {\cal P}({\cal Y})\},
\end{align}
where $P_{X,Uni}$ is the uniform distribution on ${\cal X}$.
Since $\Gamma_{\cal E}^{(m)}[W_{Y|X} \times P_X]=(W_{Y|X} \cdot P_X) \times P_{X,Uni}
$, the commitment capacity is rewritten as
\begin{align}
\log |{\cal X}|-C_c(W_{Y|X})
=&\min_{P_X}  H(X)_{P_{X,Uni}} + H(X)_{W_{Y|X} \cdot P_X}
- H(X Y)_{W_{Y|X} \times P_X}\nonumber\\
=&\min_{P_X} D( W_{Y|X} \times P_X\| (W_{Y|X} \cdot P_X) \times P_{X,Uni} ) \nonumber\\
=&\min_{P_X} D( W_{Y|X} \times P_X\| \Gamma_{\cal E}^{(m)}[W_{Y|X} \times P_X] ).\Label{MOP}
\end{align}
Hence, the minimization \eqref{MOP} is a special case of the minimization \eqref{AMR}.
Since $ \Gamma_{\cal E}^{(m)}[W_{Y|X} \times Q_X](x,y)
=(W_{Y|X} \cdot Q_X) \times P_{X,Uni} $,
\begin{align}
D(\Gamma_{\cal E}^{(m)}[P^*] \| \Gamma_{\cal E}^{(m)}[Q_X]) 
=D(W_{Y|X}\cdot P_X\|W_{Y|X}\cdot Q_X) 
\le D (P^*\| Q_X)  ,
\end{align}
which yields the condition \eqref{AMO}. Hence, the global convergence is guaranteed.

By applying \eqref{AMD}, ${\cal F}_3$ is calculated as
\begin{align}
&{\cal F}_3[W_{Y|X} \times Q_X](x,y)\nonumber\\
=&\frac{1}{\kappa_{W_{Y|X}}^{1}[Q_X]}
w_{Y|X}(y|x)^{\frac{\gamma-1}{\gamma}}  
Q_X(x)^{\frac{\gamma-1}{\gamma}} 
(w_{Y|X} \cdot Q_X)(y)^{\frac{1}{\gamma}}  P_{X,Uni}(x)^{\frac{1}{\gamma}},\Label{AMD2}
\end{align}
where $\kappa_{W_{Y|X}}^{1}[Q_X]$ is the normalizer.
Then, 
after a complicated calculation,
the marginal distribution of its projection to ${\cal M}_a$ is given as
\begin{align}
&\int_{{\cal Y}} \Gamma_{{\cal M}_a}^{(m)}[{\cal F}_3[W_{Y|X} \times Q_X]](x,y)\mu(dy)\nonumber\\
=&\frac{1}{\kappa_{W_{Y|X}}^{2}[Q_X]} Q_X(x)^{1-\frac{1}{\gamma}} \exp (  -\frac{1}{\gamma}  
D(W_{Y|X=x}\| W_{Y|X} \cdot Q_X)) ,\Label{MKD}
\end{align}
where $\kappa_{W_{Y|X}}^{2}[Q_X]$ is the normalizer.
\if0
the normalizer $\kappa[P_{Y|X} \times P_X]$ is given as
\begin{align}
&\kappa[P_{Y|X} \times P_X]\\
:=&\sum_{x\in {\cal X}}
 \int_{{\cal Y}}
p_{Y|X}(y|x)^{\frac{\gamma-1}{\gamma}}  
P_X(x)^{\frac{\gamma-1}{\gamma}} 
(p_{Y|X} \cdot P_X)(y)^{\frac{1}{\gamma}}  P_{X,Uni}(x)^{\frac{1}{\gamma}}
\mu(dy).
\end{align}
Then, 
\begin{align}
D(P_{Y|X}\times Q_X \| {\cal F}_3[P_{Y|X} \times P_X])
\end{align}
\fi
In the algorithm,
we update 
$P_X^{(t+1)}$ as
$P_X^{(t+1)}(x):=  \int_{{\cal Y}}
  \Gamma_{{\cal M}_a}^{(m)}[{\cal F}_3
  [W_{Y|X} \times P_X^{(t)}]]  (x,y)\mu(dy) $.

\subsection{Direct Application}
The update formula \eqref{MKD} requires a complicated calculation, 
we can derive the same update rule by a simpler derivation as follows.
The commitment capacity is rewritten as
\begin{align}
-C_c(W_{Y|X})
=&\min_{P_X}  I(P_X,W_{Y|X})-H(X)_{P_{X}} \nonumber\\
=&\min_{P_X} \sum_{x\in {\cal X}}P_X(x) (D(W_{Y|X=x}\| W_{Y|X}\cdot P_X )+\log P_X(x)).
\end{align}
We choose ${\cal M}_a$ and  $\Psi$ as ${\cal P}({\cal X})$ and 
\begin{align}
\Psi_{c,W_{Y|X}}[P_X](x):= D(W_{Y|X=x}\| W_{Y|X}\cdot P_X )+\log P_X(x).
\end{align}
Then, we have
\begin{align}
&\sum_{x \in {\cal X}}P_X(x)( \Psi[P_X](x)- \Psi[Q_X](x))\nonumber\\
=& D(P_X\|Q_X)-D(W_{Y|X}\cdot P_X\|W_{Y|X}\cdot Q_X) \ge 0
\end{align}
and 
\begin{align}
D(P_X\|Q_X) \ge D(P_X\|Q_X)- D(W_{Y|X}\cdot P_X\|W_{Y|X}\cdot Q_X) .
\end{align}
Since the condition (A1) with $\gamma=1$ and the condition (A2) hold,
Algorithm \ref{AL1} converges with $\gamma =1$.
In this case, ${\cal F}_3$ is given as 
\begin{align}
&{\cal F}_3[Q_X](x) \nonumber\\
&=\frac{1}{\kappa_{W_{Y|X}}^3[Q_X]}Q_X(x) \exp ( - \frac{1}{\gamma}  
(\log Q_X(x)+
D(W_{Y|X=x}\| W_{Y|X} \cdot Q_X))) \nonumber\\
&=\frac{1}{\kappa_{W_{Y|X}}^3[Q_X]}Q_X(x)^{1-\frac{1}{\gamma}} \exp (  -\frac{1}{\gamma}  
D(W_{Y|X=x}\| W_{Y|X} \cdot Q_X)) ,
\end{align}
where the normalizing constant $\kappa_{W_{Y|X}}^3[Q_X]$ is given as
$\kappa_{W_{Y|X}}^3[Q_X]:=$\par\noindent
$\sum_{x \in {\cal X}}
Q_X(x)^{1-\frac{1}{\gamma}} \exp ( - \frac{1}{\gamma}  
D(W_{Y|X=x}\| W_{Y|X} \cdot Q_X))$.
Since ${\cal F}_3[Q_X] \in {\cal M}_a$, 
$P_X^{(t+1)}$ is given as ${\cal F}_3[P_X^{(t)}]$.

To consider the effect of the acceleration parameter $\gamma$,
we made a numerical comparison when the channel with 
${\cal X}=\{1,2,3,4\}$ and ${\cal Y}=\{1,2,3,4\}$ is given as follows.
\begin{align}
&W_{Y|X}(1,1)=0.6,
W_{Y|X}(2,1)=0.2,
W_{Y|X}(3,1)=0.1,
W_{Y|X}(4,1)=0.1, \nonumber \\
&W_{Y|X}(1,2)=0.1,
W_{Y|X}(2,2)=0.2,
W_{Y|X}(3,2)=0.1,
W_{Y|X}(4,2)=0.6,\nonumber  \\
&W_{Y|X}(1,3)=0.1,
W_{Y|X}(2,3)=0.2,
W_{Y|X}(3,3)=0.15,
W_{Y|X}(4,3)=0.55,\nonumber  \\
&W_{Y|X}(1,4)=0.05,
W_{Y|X}(2,4)=0.85,
W_{Y|X}(3,4)=0.05,
W_{Y|X}(4,4)=0.05
\Label{NZU}.
\end{align}
We choose $\gamma$ to be $1$, $0.95$, and $0.9$.
Fig. \ref{con-fig1} 
shows the numerical result for the iteration of our algorithm
when the channel input is limited into $\{1,2,3\}$.
In this case, the improvement by a smaller $\gamma$ is negligible
%does not improve the convergence.
Fig. \ref{con-fig2} 
shows the same numerical result 
when all elements of $\{1,2,3,4\}$ are allowed as
the channel input.
In this case, a smaller $\gamma$ improves the convergence.

\begin{figure}[htbp]
%\centering
%\includegraphics[scale=0.4]{MHRepeater.png}
%\begin{minipage}[b]{0.49\linewidth}
    \centering
    \includegraphics[keepaspectratio, scale=0.6]{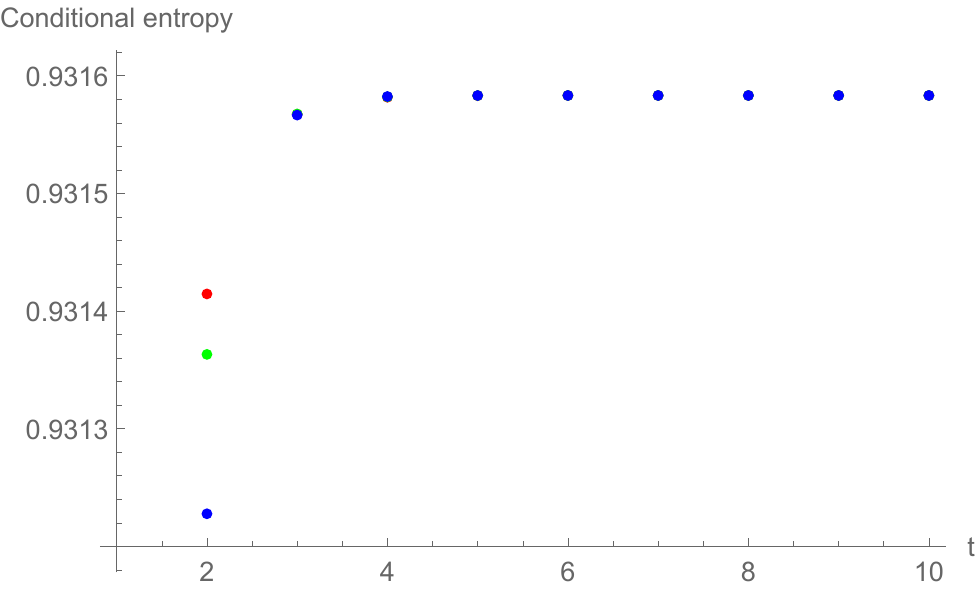}
%  \end{minipage}
%  \begin{minipage}[b]{0.49\linewidth}
    \centering
    \includegraphics[keepaspectratio, scale=0.6]{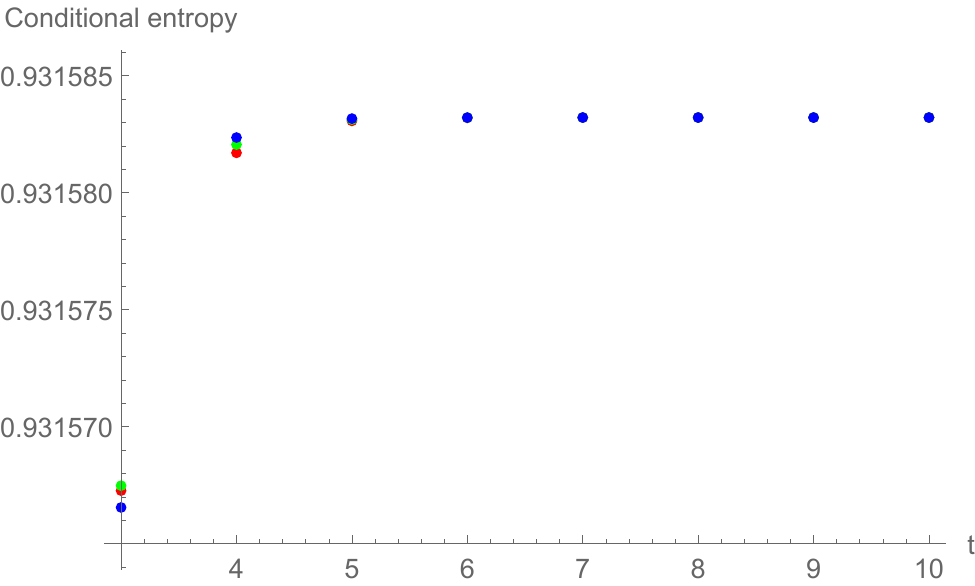}
%  \end{minipage}
\caption{
Calculation of commitment capacity for the channel given in \eqref{NZU}
with ${\cal X}=\{1,2,3\}$.
The right plot shows an enlarged plot of the left plot.
The horizontal axis shows the number of iterations.
The vertical axis shows the conditional entropy.
Red points show the case with $\gamma=1$.
Green points show the case with $\gamma=0.95$.
Blue points show the case with $\gamma=0.9$.
For $t=5,6,\ldots, 10$, these cases have almost the same value.
Hence, these plots cannot be distinguished for $t=5,6,7,8,9,10$.
At $t=2,3$, the case with $\gamma =1$ is better than other cases.
However, 
in this case, a smaller $\gamma$ does not improve the convergence.
%In this case, a smaller $\gamma$ does not improve the convergence.
}
\Label{con-fig1}
\end{figure}   

\begin{figure}[htbp]
%\centering
%\includegraphics[scale=0.4]{MHRepeater.png}
%\begin{minipage}[b]{0.49\linewidth}
    \centering
    \includegraphics[keepaspectratio, scale=0.6]{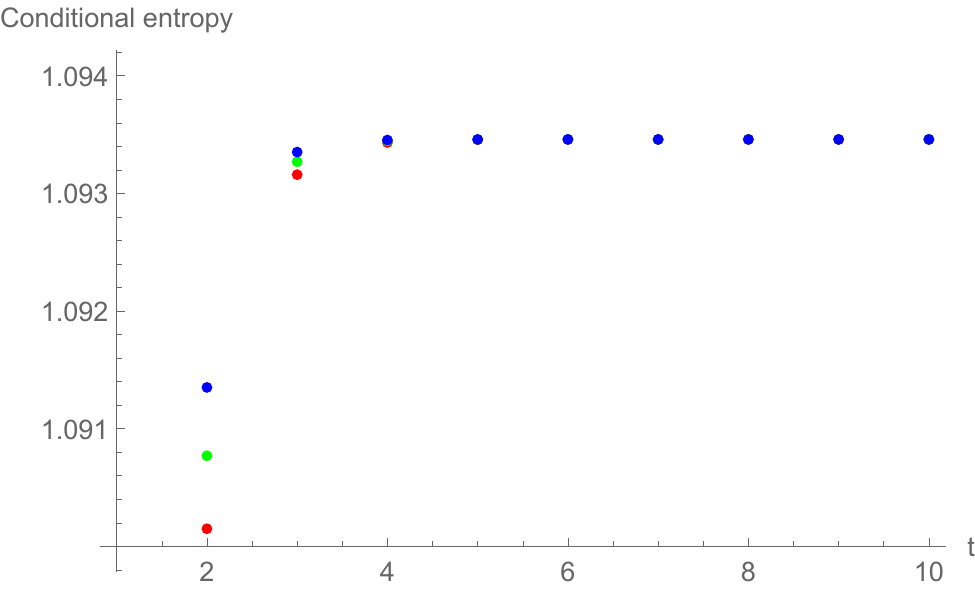}
%  \end{minipage}
%  \begin{minipage}[b]{0.49\linewidth}
    \centering
    \includegraphics[keepaspectratio, scale=0.6]{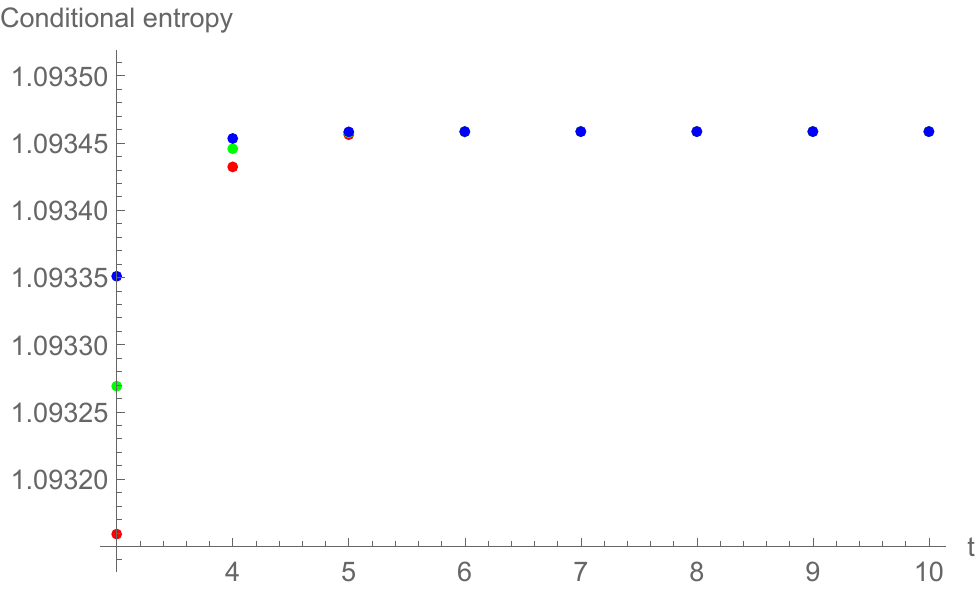}
%  \end{minipage}
\caption{
Calculation of commitment capacity for the channel given in \eqref{NZU}
with ${\cal X}=\{1,2,3,4\}$.
The role of color is the same as Fig. \ref{con-fig1}.
In this case, a smaller $\gamma$ improves the convergence.
}
\Label{con-fig2}
\end{figure}   

\newpage
\section{Reverse em problem}\Label{S7}
\subsection{General problem description}
In this section, given a pair of an exponential family ${\cal E}$ and 
a mixture family ${\cal M}_a$ on ${\cal X}$,
we consider the following maximization;
\begin{align}
&\max_{P\in {\cal M}_a}\min_{Q\in {\cal E}} D(P\|Q)=
\max_{P\in {\cal M}_a} D(P\|\Gamma_{\cal E}^{(m)}[P]) \nonumber \\
=&
\max_{P\in {\cal M}_a} \sum_{x \in {\cal X}}P(x)
(\log P(x) - \log \Gamma_{\cal E}^{(m)}[P](x))\Label{AMR2}
\end{align}
while Section \ref{S5} considers the minimization of the same value.
When ${\cal M}_a$ is given as \eqref{CPA} and ${\cal E}$ is given as
${\cal P}({\cal X})\times {\cal P}({\cal Y})$,
this problem coincides with the channel capacity \eqref{CMD}.
This problem was firstly studied for the channel capacity in \cite{Shoji}, and
was discussed with a general form in \cite{reverse}. 
To discuss this problem,
we choose the function $\Psi$ as $\Psi_{\rm{rem}}:=-\Psi_{\rm{em}}$,
and apply the discussion in Section \ref{setup}.
Due to \eqref{ZME}, \eqref{BK1} in the condition (A1) is written as
\begin{align}
(\gamma+1) D (P^0\| Q) 
\ge D(\Gamma_{\cal E}^{(m)}[P^0] \| \Gamma_{\cal E}^{(m)}[Q]),
\end{align}
and
\eqref{XMZ} in the condition (A2) is written as
\begin{align}
D(\Gamma_{\cal E}^{(m)}[P^0] \| \Gamma_{\cal E}^{(m)}[Q])
\ge D (P^0\| Q) .
\end{align}
Further, due to Lemma \ref{L9}, 
the condition (A4) holds.

\subsection{Application to wiretap channel}\Label{MXT}
Now, we apply this problem setting to wiretap channel with the degraded case 
discussed in Subsection \ref{S4-5-2}.
We choose ${\cal M}_a$ as 
$\{  W_{YZ|X}\times P_X |  P_X\in {\cal P}({\cal X}) \}$
and ${\cal E}$ as the set of
distributions with the Markov chain condition $X-Z-Y$ \cite{Toyota}.
Then, the conditional mutual information $I(X;Y|Z)[P_X,W_{YZ|X}]$ is given as
$D( W_{YZ|X}\times P_X\| \Gamma_{{\cal E}}^{(m)} (W_{YZ|X}\times P_X))$.
In this application, we have
\begin{align}
D_{\Psi_{\rm{rem}}}(W_{YZ|X}\times P_X\| W_{YZ|X}\times Q_X)&=
D_{\Psi_{W_{YZ|X}}}D( P_X \| Q_X)
\\
D( W_{YZ|X}\times P_X\| W_{YZ|X}\times Q_X) &=D( P_X \| Q_X).
\end{align}
To check the conditions (A1) and (A2) for $\Psi_{W_{YZ|X}}$,
it is sufficient to check them for $\Psi_{\rm{rem}}$ in this application.
Since we have
\begin{align}
& D( \Gamma_{{\cal E}}^{(m)} (W_{YZ|X}\times P_X)\| 
\Gamma_{{\cal E}}^{(m)} (W_{YZ|X}\times Q_X)) \nonumber\\
=& D( W_{Z|X} \times P_{X} \| Q_{XZ})
+ D( W_{YZ|X}\cdot P_X \| W_{YZ|X}\cdot Q_X)\nonumber\\
&- D( W_{Z|X}\cdot P_X \| W_{Z|X}\cdot Q_X) \nonumber\\
=& D( P_X \| Q_X)
+ D( W_{YZ|X}\cdot P_X \| W_{YZ|X}\cdot Q_X)%\nonumber\\
- D( W_{Z|X}\cdot P_X \| W_{Z|X}\cdot Q_X) \nonumber\\
\le & 2D( P_X \| Q_X),
\end{align}
LHS of \eqref{BK1} in the condition (A1) is written as
\begin{align}
& \gamma D (P_X^0\| Q_X) 
- D( W_{YZ|X}\cdot P_X \| W_{YZ|X}\cdot Q_X)
+ D( W_{Z|X}\cdot P_X \| W_{Z|X}\cdot Q_X) \nonumber\\
\ge &
\gamma D (P^0\| Q) - D( P_X^0 \|  Q_X).
\end{align}
It is not negative when $\gamma \ge 1$.
Also, 
RHS of \eqref{XMZ} in the condition (A2) is written as
\begin{align}
D( W_{YZ|X}\cdot P_X \| W_{YZ|X}\cdot Q_X)
- D( W_{Z|X}\cdot P_X \| W_{Z|X}\cdot Q_X) 
\ge 0.
\end{align}
Hence, the conditions (A1) and (A2) hold
with $\gamma \ge 1$.

\section{Information bottleneck}\Label{S7B}
As a method for information-theoretical machine learning,
we often consider information bottleneck \cite{Tishby}.
Consider two correlated systems ${\cal X}$ and ${\cal Y}$
and a joint distribution $P_{XY}$ over ${\cal X} \times {\cal Y}$.
The task is to extract an essential information
from the space ${\cal X}$ to ${\cal T}$ with respect to ${\cal Y}$.
Here, we discuss a generalized problem setting proposed in \cite{Strouse}.
For this information extraction, given parameters $\alpha \in [0,1]$ and $\beta
\ge \alpha$,
we choose a conditional distribution
$P_{T|X}^*$ as
\begin{align}
P_{T|X}^*:= \argmin_{P_{T|X}} 
\alpha I(T;X)+(1-\alpha)H(T)-\beta I(T;Y).
\end{align}
This method is called information bottleneck.
To apply our method to this problem, 
we choose ${\cal M}_a$ ${\cal P}({\cal X})$, 
and define
\begin{align}
{\Psi}_{\alpha,\beta}[P_{TX}](t,x)
:=
&\alpha \log P_{TX}(t,x)
-\alpha \log P_{X}(x)
+(\beta-1) \log P_{T}(t) \nonumber \\
&+\beta \sum_{y\in {\cal Y}} P_{Y|X}(y|x) (
%\log P_{T}(t) + 
\log P_{Y}(y) 
- \log P_{TY}(t,y) ).
\end{align}
Then, 
when 
the joint distribution $P_{TX}$ is chosen to be $P_{T|X} \times P_X$,
the objective function is written as
\begin{align}
{\cal G}_{\alpha,\beta}( P_{TX}):=
\sum_{t\in {\cal T},x \in {\cal X}}
P_{TX}(t,x) {\Psi}_{\alpha,\beta}[P_{TX}](t,x).
\end{align}
That is, our problem is reduced to the minimization
\begin{align}
\min_{P_{TX}\in {\cal M}(P_X)}
{\cal G}_{\alpha,\beta}( P_{TX}),
\end{align}
where ${\cal M}(P_X)$ is the set of joint distributions on 
${\cal X} \times {\cal Y}$
whose marginal distribution on ${\cal X}$ is $P_X$.
The $e$-projection $\Gamma^{(e)}_{{\cal M}(P_X)}$
to ${\cal M}(P_X)$ is written as
\begin{align}
\Gamma^{(e)}_{{\cal M}(P_X)}[Q_{TX}]=
Q_{T|X}\times P_X\Label{MZY}
\end{align}
because the relation 
$D(P_{TX}\| Q_{TX})=
D(P_{X}\| Q_{X})+
D(P_{TX}\| Q_{T|X}\times P_X)$ holds for a distribution
$P_{TX}\in {\cal M}(P_X)$.

When Algorithm \ref{AL1} is applied to this problem,
due to \eqref{MZY},
the update rule for the conditional distribution
is given as
\begin{align}
P_{T| X}^{(t)}\mapsto 
P_{T| X}^{(t+1)}(t|x)
:= \kappa_x 
P_{T|X}^{(t)}(t|x) \exp( -\frac{1}{\gamma}\Psi_{\alpha,\beta}[
P_{T|X}^{(t)}\times P_X ](t,x)),
\end{align}
where $\kappa_x$ is a normalized constant.
This update rule is the same as the update rule proposed in 
Section 3 of \cite{HY}
when the states $\{\rho_{Y|x}\}$ are given as diagonal density matrices,
i.e., a conditional distribution $P_{Y|X}$.
Also, as shown in \cite[(22)]{HY},
we have the relation 
\begin{align}
D_{\Psi_{\alpha,\beta}}( P_{TX} \|Q_{TX})
\le \alpha D( P_{TX} \|Q_{TX})
\end{align}
for $P_{TX},Q_{TX} \in {\cal M}(P_X)$ as follows.
First, we have
\begin{align}
D( P_{T|X} \cdot P_{XY}\|Q_{T|X}\cdot P_{XY})
\ge 
D( P_{T} \|Q_{T}),
\end{align}
where
$P_{T|X} \cdot P_{XY}(t,y):=
\sum_{x \in {\cal X}} P_{T|X}(t|x) P_{XY}(x,y)$.
Then, we have
\begin{align}
& D_{\Psi_{\alpha,\beta}}( P_{TX} \|Q_{TX}) \nonumber \\
=& (\beta-1)D( P_{T} \|Q_{T})
+\alpha \sum_{x \in {\cal X}}P_X(x)D( P_{T|X=x}\|Q_{T|X=x})\nonumber  \\
&
-\beta D( P_{T|X} \cdot P_{XY}\|Q_{T|X}\cdot P_{XY})
\nonumber \\
\le & -D( P_{T} \|Q_{T})
+\alpha \sum_{x \in {\cal X}}P_X(x)D( P_{T|X=x}\|Q_{T|X=x})\nonumber \\
\le &
\alpha \sum_{x \in {\cal X}}P_X(x)D( P_{T|X=x}\|Q_{T|X=x}) 
= \alpha D( P_{TX} \|Q_{TX}).
\end{align}
That is,
the condition (A1) holds with $\gamma=\alpha$.
Therefore, with $\gamma=\alpha$,
Theorem \ref{TTH1} guarantees that 
Algorithm \ref{AL1} converges to a local minimum,
which was shown as \cite[Theorem 3]{HY}.
This fact shows the importance of the choice of 
$\gamma$ dependently on the problem setting.
That is, it shows the necessity of 
our problem setting with a general positive real number $\gamma>0$.

The paper \cite{HY} also discussed the case when 
${\cal Y}$ and ${\cal T}$ are quantum systems.
It numerically compared these algorithms depending on $\gamma$ \cite[Fig. 2]{HY}.
This numerical calculation indicates the following behavior.
When $\gamma$ is larger than a certain threshold, 
a smaller $\gamma$ realizes faster convergence.
But, when $\gamma$ is smaller than a certain threshold, 
the algorithm does not converge.

\section{Conclusion}
We have proposed iterative algorithms with an acceleration parameter for
a general minimization problem over a mixture family.
For these algorithms, we have shown convergence theorems in various forms,
one of which covers the case with approximated iterations.
Then, we have applied our algorithms to various problem settings including 
the em algorithm and several information theoretical problem settings.

There are two existing studies to numerically evaluate the effect of 
the acceleration parameter $\gamma$ \cite{RISB,HY}.
They reported improvement in the convergence by modifying 
the acceleration parameter $\gamma$.
For example, in the numerical calculation for information bottleneck
 in \cite[Fig. 2]{HY},
the case with $\gamma=0.55 $ improves the convergence.
Our numerical calculation for the commitment capacity
has two cases.
In one case, the choices with $\gamma=0.95,0.9$ do not improve the convergence.
In another case, the choices with $\gamma=0.95,0.9$ improve the convergence.
These facts show that the effect of the acceleration parameter $\gamma$ 
depends on the parameters of the problem setting.
The commitment capacity is considered as a special case of 
the divergence between a mixture family and an exponential family.
\if0
Therefore, it is useful to characterize the case when 
the acceleration parameter $\gamma$
improves the conventional em algorithm for the 
minimization of the divergence between a mixture family and an exponential family
in the problem setting of the em algorithm.
\fi

There are several future research directions.
The first direction is the evaluation of the convergence speed 
of Algorithm \ref{AL4} because we could not derive its evaluation.
The second direction is to find various applications of our methods.
Although this paper studied several examples, 
it is needed to more useful examples for our algorithm.
The third direction is the extensions of our results.
A typical extension is the extension to the quantum setting \cite{Holevo,SW,hayashi}.
As a further extension, it is an interesting topic to extend our result to 
the setting with Bregman divergence.
Recently, Bregman proximal gradient algorithm has been studied for the minimization of 
a convex function \cite{CT93,T97,ZYS}.
Since this algorithm uses Bregman divergence,
it might have an interesting relation with the above-extended algorithm.
Therefore, it is an interesting study to investigate this relation.

\section*{Acknowledgments}
The author was supported in part by the National Natural Science Foundation of China (Grant No. 62171212) and
Guangdong Provincial Key Laboratory (Grant No. 2019B121203002).
The author is very grateful to Mr. Shoji Toyota for helpful
discussions. % and explaining the achievements of the reference \cite{Shoji}.
%In particular, he explained the author what problems were not solved in the reference \cite{Shoji}.
In addition, he pointed out that 
the secrecy capacity can be written as the reverse em algorithm in a similar way as the channel capacity \cite{Toyota}
under the degraded condition.

\section*{Data availability}
Data sharing is not applicable to this article as no datasets were generated
or analyzed during the current study.

\appendix
%\section{Proofs of main results}\Label{S3}
\section{Useful lemma}\Label{S3-1}
To show various theorems, we prepare the following lemma.
\begin{lemma}\Label{LLX}
For any two distributions $Q ,Q'\in {\cal M}_a$, we have
\begin{align}
&D(P^{0}\| Q )- D(P^{0}\| Q')  \nonumber\\
=&\frac{1}{\gamma}J_\gamma(\Gamma^{(e)}_{{\cal M}_a}[{\cal F}_3[Q]],Q) 
-\frac{1}{\gamma}{\cal G}(P^{0})
+\frac{1}{\gamma}D_{\Psi}(P^{0} \|Q)
%\sum_{x \in {\cal X}} P^{0}(x)  
%\Big(\Psi[P^{0}](x)-\Psi[Q](x)\Big)
%\nonumber\\ &
-D(\Gamma^{(e)}_{{\cal M}_a}[{\cal F}_3[Q]]\| Q' )\Label{XM1} \\
=&\frac{1}{\gamma}{\cal G}(\Gamma^{(e)}_{{\cal M}_a}[{\cal F}_3[Q]]) -\frac{1}{\gamma}{\cal G}(P^{0}) 
%\nonumber \\ &
+
D(\Gamma^{(e)}_{{\cal M}_a}[{\cal F}_3[Q]]\|Q)
-\frac{1}{\gamma} 
D_{\Psi}(\Gamma^{(e)}_{{\cal M}_a}[{\cal F}_3[Q]] \|Q)
%\sum_{x \in {\cal X}} \Gamma^{(e)}_{{\cal M}_a}[{\cal F}_3[Q]](x) 
%( \Psi[\Gamma^{(e)}_{{\cal M}_a}[{\cal F}_3[Q]]](x)-\Psi[Q](x))
\nonumber\\
&+\frac{1}{\gamma}
D_{\Psi}(P^{0} \|Q)
%\sum_{x \in {\cal X}} P^{0}(x)  \Big(\Psi[P^{0}](x)-\Psi[Q](x)\Big)
-D(\Gamma^{(e)}_{{\cal M}_a}[{\cal F}_3[Q]]\| Q' ) .\Label{XM2}
\end{align}
In addition, when $\Psi$ is defined for any distribution in ${\cal P}({\cal X})$,
the above relations holds for any distribution $Q \in {\cal P}({\cal X})$.
\end{lemma}

\begin{proof}
We have
\begin{align}
{\cal G}(P^{0})=&
\sum_{x \in {\cal X}} P^{0}(x)  \Psi[P^{0}](x) % \nonumber\\
=J _\gamma(\Gamma^{(e)}_{{\cal M}_a}[{\cal F}_3[P^{0}]],P^{0})\nonumber\\
=& \gamma (D(\Gamma^{(e)}_{{\cal M}_a}[{\cal F}_3[P^{0}]]\|
{\cal F}_3[P^{0}])
- \log \kappa[P^{0}]).
\Label{MLP}
\end{align}
Using \eqref{MLP}, we have
\begin{align}
&D(P^{0}\| Q)- D(P^{0}\| Q')  %\nonumber\\
=\sum_{x \in {\cal X}} P^{0}(x) (\log Q'(x)- \log Q(x))\nonumber \\
=&\sum_{x \in {\cal X}} P^{0}(x) \Big(
\log Q'(x)
-\log {\cal F}_3[Q](x)
+\log {\cal F}_3[Q](x)
- \log Q(x) \Big) \nonumber\\
\stackrel{(a)}{=}&
D(P^{0}\| {\cal F}_3[Q])
-D(P^{0}\| Q' ) % \nonumber\\
+\sum_{x \in {\cal X}} P^{0}(x) \Big(
-\frac{1}{\gamma} \Psi[Q](x) -\log \kappa[Q] \Big)\nonumber \\
\stackrel{(b)}{=}&
D(\Gamma^{(e)}_{{\cal M}_a}[{\cal F}_3[Q]]\| {\cal F}_3[Q])
-D(\Gamma^{(e)}_{{\cal M}_a}[{\cal F}_3[Q]]\| Q' ) \nonumber\\
&-\log \kappa[Q] 
-\frac{1}{\gamma}\sum_{x \in {\cal X}} P^{0}(x)  \Psi[Q](x)\nonumber\\
\stackrel{(c)}{=}&
\frac{1}{\gamma}J_\gamma(\Gamma^{(e)}_{{\cal M}_a}[{\cal F}_3[Q]],Q) -\frac{1}{\gamma}{\cal G}(P^{0})
+\frac{1}{\gamma}\sum_{x \in {\cal X}} P^{0}(x)  
\Big(\Psi[P^{0}](x)-\Psi[Q](x)\Big)\nonumber\\
&-D(\Gamma^{(e)}_{{\cal M}_a}[{\cal F}_3[Q]]\| Q' ) \nonumber\\
\stackrel{(d)}{=}&
\frac{1}{\gamma}{\cal G}(\Gamma^{(e)}_{{\cal M}_a}[{\cal F}_3[Q]]) -\frac{1}{\gamma}{\cal G}(P^{0})\nonumber\\
&+D(\Gamma^{(e)}_{{\cal M}_a}[{\cal F}_3[Q]]\|Q)
-\frac{1}{\gamma}\sum_{x \in {\cal X}} \Gamma^{(e)}_{{\cal M}_a}[{\cal F}_3[Q]](x) 
( \Psi[\Gamma^{(e)}_{{\cal M}_a}[{\cal F}_3[Q]]](x)-\Psi[Q](x))\nonumber\\
&+\frac{1}{\gamma}\sum_{x \in {\cal X}} P^{0}(x)  
\Big(\Psi[P^{0}](x)-\Psi[Q](x)\Big)
-D(\Gamma^{(e)}_{{\cal M}_a}[{\cal F}_3[Q]]\| Q' ) ,
\end{align}
where each step is shown as follows.
$(a)$ follows from the definition of ${\cal F}_3(Q) $. 
$(c)$ follows from \eqref{XMY} and \eqref{MLP}. 
$(d)$ follows from \eqref{XMY5}. 
$(b)$ follows from the relations 
\begin{align}
D(P^{0}\| {\cal F}_3[Q])&=D(P^{0}\| \Gamma^{(e)}_{{\cal M}_a}[{\cal F}_3[Q]])
+D(\Gamma^{(e)}_{{\cal M}_a}[{\cal F}_3[Q]]\| {\cal F}_3[Q]) \\
D(P^{0}\| Q' ) &=D(P^{0}\| \Gamma^{(e)}_{{\cal M}_a}[{\cal F}_3[Q]] )+D(\Gamma^{(e)}_{{\cal M}_a}[{\cal F}_3[Q]]\| Q' ), 
\end{align}
which are shown by the Phythagorean equation.
Therefore, considering the definition of 
$D_\Psi(P\|Q)$,
we obtain \eqref{XM1} and \eqref{XM2}.
\end{proof}

\section{Proof of Theorem \ref{TH1} and Corollary \ref{Cor1}}\Label{S3-2}
\noindent{\bf Step 1:}\quad
This step aims to show the following inequalities
%First, we show the case (ii) 
by assuming that item (i) does not hold and the conditions (A1) and (A2) hold.
\begin{align}
&D(P^{0}\| P^{(t+1)}) \le \delta \Label{XPZ} \\
&D(P^{0}\| P^{(t)})- D(P^{0}\| P^{(t+1)}) \ge \frac{1}{\gamma}{\cal G}(P^{(t+1)}) -\frac{1}{\gamma}{\cal G}(P^{0})\Label{XPZ2}
\end{align}
for $t=1, \ldots, t_0-1$.
We show these relations by induction.

For any $t$, 
by using the relation $ \Gamma^{(e)}_{{\cal M}_a}[{\cal F}_3[{P}^{(t)}]]=P^{(t+1)}$, 
the application of \eqref{XM2} of Lemma \ref{LLX} to the case with $Q=P^{(t)}$ and $Q'=P^{(t+1)}$
yields
\begin{align}
&D(P^{0}\| P^{(t)})- D(P^{0}\| P^{(t+1)})  \nonumber\\
=&\frac{1}{\gamma}{\cal G}( P^{(t+1)}) 
-\frac{1}{\gamma}{\cal G}(P^{0}) %\nonumber\\
+D(\Gamma^{(e)}_{{\cal M}_a}[{\cal F}_3[P^{(t)}]]\|P^{(t)})\nonumber\\
&-\frac{1}{\gamma}
D_\Psi( \Gamma^{(e)}_{{\cal M}_a}[{\cal F}_3[P^{(t)}]] \| \Psi[P^{(t)}])
%\sum_{x \in {\cal X}} \Gamma^{(e)}_{{\cal M}_a}[{\cal F}_3[P^{(t)}]](x) 
%( \Psi[\Gamma^{(e)}_{{\cal M}_a}[{\cal F}_3[P^{(t)}]](x)-\Psi[P^{(t)}](x))
%\nonumber\\ &
+\frac{1}{\gamma}
D_\Psi(P^{0}\|P^{(t)})
%\sum_{x \in {\cal X}} P^{0}(x)  
%\Big(\Psi[P^{0}](x)-\Psi[P^{(t)}](x)\Big)
%-D(\Gamma^{(e)}_{{\cal M}_a}[{\cal F}_3[P^{(t)}]]\| P^{(t+1)} ) 
.\Label{XME2}
\end{align}

First, we show the relations \eqref{XPZ} and \eqref{XPZ2} with $t=1$.
Since $D(P^{0}\| P^{(1)}) \le \delta$,
$P^{(1)}$ belongs to $U(P^{0},\delta)$. Hence,
the conditions (A1) and (A2) guarantee the following inequality with $t=1$;
%the relations \eqref{BK1} and \eqref{XMZ} holds
%we have
\begin{align}
\hbox{\rm(RHS of \eqref{XME2})} 
\ge %\stackrel{(a)}{\ge} &
\frac{1}{\gamma}{\cal G}(P^{(t+1)}) -\frac{1}{\gamma}{\cal G}(P^{0}).\Label{AXE}
\end{align}
The combination of \eqref{XME2} and \eqref{AXE} implies \eqref{XPZ2}.
Since item (i) does not hold, we have
\begin{align}
\frac{1}{\gamma}{\cal G}(P^{(t+1)}) -\frac{1}{\gamma}{\cal G}(P^{0})\ge 0\Label{AMQ}.
\end{align}
The combination of \eqref{XME2}, \eqref{AXE}, and \eqref{AMQ} implies \eqref{XPZ}.

Next, we show the relations \eqref{XPZ} and \eqref{XPZ2} with $t=t'$
by assuming the relations \eqref{XPZ} and \eqref{XPZ2} with $t=t'-1$.
Since the assumption guarantees $D(P^{0}\| P^{(t')}) \le \delta$,
the conditions (A1) and (A2) guarantee \eqref{AXE} with $t=t'$.
We obtain \eqref{XPZ} and \eqref{XPZ2} in the same way as $t=1$.

\noindent{\bf Step 2:}\quad
This step aims to show \eqref{XME}
%First, we show the case (ii) 
by assuming that item (i) does not hold and the conditions (A1) and (A2) hold.
Due to \eqref{XPZ}, the condition (A1) and Lemmas \ref{L1} and \ref{L2}
guarantee that
\begin{align}
{\cal G}(P^{(t+1)}) \le {\cal G}(P^{(t)}) \Label{AMK}.
\end{align}
We have
\begin{align}
& \frac{t_0}{\gamma} \Big(%\min_{t=2,\ldots,k+1}{\cal G}(P^{(t)})
{\cal G}(P^{(t_0+1)}) - {\cal G}(P^{0}) \Big)
%\nonumber\\
\stackrel{(a)}{\le}  \frac{1}{\gamma}\sum_{t=1}^{t_0}
{\cal G}(P^{(t+1)}) -{\cal G}(P^{0}) %J^{(t)}-J^*
\nonumber\\
\stackrel{(b)}{\le}& \sum_{t=1}^{t_0}
D(P^{0}\| P^{(t)})- D(P^{0}\| P^{(t+1)})  %\nonumber\\
= 
D(P^{0}\| P^{(1)})-
D(P^{0}\| P^{(t_0+1)})  
\le D(P^{0}\| P^{(1)}),
\end{align}
where $(a)$ and $(b)$ follow from \eqref{AMK} and \eqref{XPZ2}, respectively.

\noindent{\bf Step 3:}\quad
This step aims to show item (ii)
by assuming the conditions (A0) as well as (A1) and (A2).
In the discussion of Step 1,
since $D(P^{0}\| P^{(t)}) \le \delta$,
the condition (A0) guarantees \eqref{AMQ}. 
We can show item (ii) with assuming that item (i) does not hold.
Hence, we obtain Theorem \ref{TH1}.

\noindent{\bf Step 4:}\quad
To show Corollary \ref{Cor1}, we apply \eqref{XME2} to the case when $P^0=P^*$ and $P^{(t)}=P^*_i$.
Then, we have
\begin{align}
0={\cal G}(P^*_i)-{\cal G}(P^*)
+\sum_{x \in {\cal X}}P^*(x) 
(\Psi[P^*](x)-\Psi[P^*_i](x)),
\end{align}
which implies \eqref{Cor2}.

\section{Proof of Theorem \ref{TH2}}\Label{S3-3}
We have already shown that 
$\{P^{(t)}\}_{t=1}^{t_0+1} 
\subset U(P^{0},\delta)$ when item (i) does not hold.
Hence, in the following, we show only \eqref{CAU2} by using (A1), (A3), and
$\{P^{(t)}\}_{t=1}^{t_0+1} \subset U(P^{0},\delta)$
when item (i) does not hold.

We have
\begin{align}
&J_\gamma(\Gamma^{(e)}_{{\cal M}_a}[{\cal F}_3[P^{(t)}]],P^{(t)})
- {\cal G}(P^{0}) \\
\stackrel{(a)}{=}&{\cal G}( P^{(t+1)}) -{\cal G}(P^{0}) %\nonumber\\
+\gamma D(\Gamma^{(e)}_{{\cal M}_a}[{\cal F}_3[P^{(t)}]]\|P^{(t)})
%\nonumber\\&
-D_\Psi(\Gamma^{(e)}_{{\cal M}_a}[{\cal F}_3[P^{(t)}]] \| \Psi[P^{(t)}])
%\sum_{x \in {\cal X}} \Gamma^{(e)}_{{\cal M}_a}[{\cal F}_3[P^{(t)}]](x) 
%( \Psi[\Gamma^{(e)}_{{\cal M}_a}[{\cal F}_3[P^{(t)}]](x)-\Psi[P^{(t)}](x))
\nonumber\\
\stackrel{(b)}{\ge}& {\cal G}( P^{(t+1)}) -{\cal G}(P^{0}) \stackrel{(c)}{\ge} 0,
\Label{PYF}
\end{align}
where $(a)$ follows from Lemma \ref{L2},
$(b)$ follows from the condition (A1) and $P^{(t)}\in  U(P^{0},\delta)$,
and
$(c)$ holds because item (i) does not hold.

Since $ \Gamma^{(e)}_{{\cal M}_a}[{\cal F}_3[{P}^{(t)}]]=P^{(t+1)}$,
the application of \eqref{XM1} of Lemma \ref{LLX} to the case with $Q=P^{(t)}$ and $Q' %_\theta
=P^{(t+1)}$
yields
\begin{align}
&  D(P^{0} \| P^{(t)})-D(P^{0} \| P^{(t+1)})\nonumber \\
=&
\frac{1}{\gamma} J_\gamma(\Gamma^{(e)}_{{\cal M}_a}[{\cal F}_3[P^{(t)}]],P^{(t)})
- \frac{1}{\gamma}{\cal G}(P^{0}) %\nonumber \\
+\frac{1}{\gamma}
D_\Psi(P^{0}\|P^{(t)})
%\sum_{x \in {\cal X}} P^{0}(x) \Big(\Psi[P^{0}](x)-\Psi[P^{(t)}](x)\Big)
\Label{MOA} \\
\stackrel{(a)}{\ge} &
\frac{1}{\gamma}
D_\Psi(P^{0}\|P^{(t)})
%\sum_{x \in {\cal X}} P^{0}(x) \Big(\Psi[P^{0}](x)-\Psi[P^{(t)}](x)\Big) %\nonumber \\
\stackrel{(b)}{\ge} 
\frac{\beta}{\gamma}D(P^{0} \| P^{(t)})\Label{MOA2},
\end{align}
where 
$(a)$ follows from \eqref{PYF}, and 
$(b)$ follows from \eqref{CAU} in the condition (A3) and $P^{(t)}\in  U(P^{0},\delta)$.
Hence, we have
\begin{align}
  D(P^{0} \| P^{(t+1)})\le (1-\frac{\beta}{\gamma})D(P^{0} \| P^{(t)}).\Label{MOA3}
\end{align}
Using the above relations, we have
\begin{align}
& {\cal G}(P^{(t+1)}) -{\cal G}(P^{0}) %\nonumber \\
\stackrel{(a)}{\le} 
J_\gamma(\Gamma^{(e)}_{{\cal M}_a}[{\cal F}_3[P^{(t)}]],P^{(t)})
-{\cal G}(P^{0}) \nonumber \\
\stackrel{(b)}{=} &D(P^{0}\| P^{(t)})- D(P^{0}\| P^{(t+1)})  -
\frac{1}{\gamma}
D_\Psi(P^{0}\|P^{(t)})
%\sum_{x \in {\cal X}} P^{0}(x)  
%\Big(\Psi[P^{0}](x)-\Psi[P^{(t)}](x)\Big)
\nonumber \\
\stackrel{(c)}{\le}  &D(P^{0}\| P^{(t)})- D(P^{0}\| P^{(t+1)})  -
\frac{\beta}{\gamma}D(P^{0}\| P^{(t)}) \nonumber \\
\stackrel{(d)}{\le} 
 & (1-\frac{\beta}{\gamma}) D(P^{0} \| P^{(t)})
\stackrel{(e)}{\le} (1-\frac{\beta}{\gamma})^t D(P^{0} \| P^{(1)}),
\end{align}
where each step is derived as follows.
Step $(a)$ follows from \eqref{PYF}.
Step $(b)$ follows from \eqref{MOA}.
Step $(c)$ follows from \eqref{CAU} in the condition (A3) and $P^{(t)}\in  U(P^{0},\delta)$.
Step $(d)$ follows from \eqref{MOA2}.
Step $(e)$ follows from \eqref{MOA3}.
Hence, we obtain \eqref{CAU2}.
Therefore, we have shown item (ii) under the conditions (A1) and (A3)
when item (i) does not hold.

When (A0) holds in addition to (A1) and (A3), 
as shown in Step 1 of the proof of Theorem \ref{TH1},
the relation $\{P^{(t)}\}_{t=1}^{t_0+1} \subset U(P^{0},\delta)$ holds.
Hence, item (ii) holds.

\section{Proof of Theorem \ref{TH8}}\Label{S3-4}
In this proof, we choose $\bar{P}^{(1)}$ to be ${P}^{(1)}$.

\noindent{\bf Step 1:}\quad
This step aims to show the inequality \eqref{XP8}.
We denote the maximizer in \eqref{AMG} by $\theta'$.
The condition \eqref{AMG} implies that 
\begin{align}
\phi[\bar{P}^{(t)}](\theta)- \sum_{j=1}^k \theta^j a_j
\le 
\phi[\bar{P}^{(t)}](\theta')- \sum_{j=1}^k {\theta'}^j a_j
+\epsilon_1 
\Label{NMA}.
\end{align}
The divergence in the exponential family $\{Q_\theta \}$
can be considered as the Bregmann divergence of the potential function
$\phi[\bar{P}^{(t)}](\theta)$.
For example, for this fact, see \cite[Section III-A]{Bregman-em}.
Hence, we have
\begin{align}
 D(\Gamma^{(e)}_{{\cal M}_a}[{\cal F}_3[\bar{P}^{(t)}]]\| \bar{P}^{(t+1)} )
%\nonumber \\
=&
\phi[\bar{P}^{(t)}](\theta)- \sum_{j=1}^k \theta^j a_j
-\Big(\phi[\bar{P}^{(t)}](\theta')- \sum_{j=1}^k {\theta'}^j a_j\Big)
\le \epsilon_1\Label{LFS4}.
\end{align}

\noindent{\bf Step 2:}\quad
This step aims to show Eq. \eqref{XZWN}
when the following inequality
\begin{align}
{\cal G}( {P}^{(t_2)})
 -\gamma D(P^{(t_2)} \| \bar{P}^{(t_2)})
%\nonumber \\
{\le}  \frac{\gamma}{t_1-1} D^F( \theta_{*} \| \theta_{(1)})+  \epsilon_1 %+\epsilon_3 
+{\cal G}(P^{0})\Label{CO8T}
\end{align}
holds.
Eq. \eqref{XZWN} is shown as follows;
\begin{align}
&{\cal G}( {P}_f^{(t_1)})-{\cal G}(P^{0}) % \nonumber \\
\stackrel{(a)}{=} 
{\cal G}( {P}^{(t_2)})-{\cal G}(P^{0}) % \nonumber \\
\stackrel{(b)}{\le} %& 
\frac{\gamma}{t_1-1} D(P^{0}\| \bar{P}^{(1)})
+\gamma D(P^{(t_2)} \| \bar{P}^{(t_2)})
+\epsilon_1 \nonumber \\
\stackrel{(b)}{\le} & 
\frac{\gamma}{t_1-1} D(P^{0}\| \bar{P}^{(1)})
+\gamma \epsilon_2+\epsilon_1, % +\epsilon_3 ,
\Label{CO8A}
\end{align}
where
Steps $(a)$, $(b)$, and $(c)$ follow from the definition of ${P}_f^{(t_1)}$, \eqref{CO8T},
and \eqref{NXP}, respectively.
Therefore, the remaining task is the proof of \eqref{CO8T}.

\noindent{\bf Step 3:}\quad
We choose $t_4 \in [1, t_1-1]$ as the minimum integer $t \in [1, t_1-1]$
to satisfy the following inequality
\begin{align}
\frac{1}{\gamma}J_\gamma(\Gamma^{(e)}_{{\cal M}_a}[{\cal F}_3[\bar{P}^{(t)}]],\bar{P}^{(t)})
\le \frac{1}{\gamma}{\cal G}(P^{0})+\epsilon_1.\Label{MKX}
\end{align}
If no integer $t \in [1, t_1-1]$ satisfies \eqref{MKX}, we set $t_4$ to be $t_1$.
This step aims to show the following two facts for $t =1, \ldots, t_4-1$.
(i) $D(P^{0}\| \bar{P}^{(t+1)}) \le \delta$. (ii) The inequality
\begin{align}
D(P^{0}\| \bar{P}^{(t)})- D(P^{0}\| \bar{P}^{(t+1)})  %\nonumber \\
{\ge} %&
\frac{1}{\gamma}J_\gamma(\Gamma^{(e)}_{{\cal M}_a}[{\cal F}_3[\bar{P}^{(t)}]],\bar{P}^{(t)})
-\frac{1}{\gamma}{\cal G}(P^{0})-\epsilon_1\Label{XMZA}
\end{align}
holds.
The above two items are shown by induction for $t$ as follows.
It is sufficient to show the case when $t_4 \ge 2$.

We show items (i) and (ii) for $t=1$ as follows.
The application of Lemma \ref{LLX} to the case with $Q=\bar{P}^{(1)}$ and $Q'%_\theta
=
\bar{P}^{(2)}$
yields
\begin{align}
&D(P^{0}\| \bar{P}^{(1)})- D(P^{0}\| \bar{P}^{(2)})  \nonumber \\
=&
\frac{1}{\gamma}J_\gamma(\Gamma^{(e)}_{{\cal M}_a}[{\cal F}_3[\bar{P}^{(1)}]],\bar{P}^{(1)})
-\frac{1}{\gamma}{\cal G}(P^{0})%\nonumber \\
%&+\gamma D(\Gamma^{(e)}_{{\cal M}_a}[{\cal F}_3[\bar{P}^{(t)}]]\|\bar{P}^{(t)})\\
%&+\sum_{x \in {\cal X}} \Gamma^{(e)}_{{\cal M}_a}[{\cal F}_3[\bar{P}^{(t)}]](x) 
%( \Psi[\Gamma^{(e)}_{{\cal M}_a}[{\cal F}_3[\bar{P}]]](x)-\Psi[\bar{P}^{(t)}](x))\\
%&
+\frac{1}{\gamma}
D_\Psi(P^{0}\|P^{(1)})
%\sum_{x \in {\cal X}} P^{0}(x)  \Big(\Psi[P^{0}](x)-\Psi[\bar{P}^{(1)}](x)\Big)
-D(\Gamma^{(e)}_{{\cal M}_a}[{\cal F}_3[\bar{P}^{(1)}]]\| \bar{P}^{(2)} ) 
\nonumber \\
\stackrel{(a)}{\ge} &
\frac{1}{\gamma}J_\gamma(\Gamma^{(e)}_{{\cal M}_a}[{\cal F}_3[\bar{P}^{(1)}]],\bar{P}^{(1)})
-\frac{1}{\gamma}{\cal G}(P^{0})-\epsilon_1
\stackrel{(b)}{\ge} 0,
\end{align}
where $(a)$ follows from (A2+) and \eqref{XP8}
because the relation $D(P^{0}\| \bar{P}^{(1)}) \le \delta$ follows from the assumption of this theorem.
$(b)$ follows from the fact that $t=1$ does not satisfy the condition \eqref{MKX}.
Hence, $ D(P^{0}\| \bar{P}^{(2)}) \le D(P^{0}\| \bar{P}^{(1)}) \le \delta$.

Assume that items (i) and  (ii) hold with $t=t'-1$.
Then, the application of Lemma \ref{LLX} to the case with $Q=\bar{P}^{(t)}$ and 
$Q'%_\theta
= \bar{P}^{(t+1)}$
yields
\begin{align}
&D(P^{0}\| \bar{P}^{(t')})- D(P^{0}\| \bar{P}^{(t'+1)})  \nonumber \\
=&
\frac{1}{\gamma}J_\gamma(\Gamma^{(e)}_{{\cal M}_a}[{\cal F}_3[\bar{P}^{(t')}]],\bar{P}^{(t')})
-\frac{1}{\gamma}{\cal G}(P^{0})
%\nonumber \\
%&+\gamma D(\Gamma^{(e)}_{{\cal M}_a}[{\cal F}_3[\bar{P}^{(t)}]]\|\bar{P}^{(t)})\\
%&+\sum_{x \in {\cal X}} \Gamma^{(e)}_{{\cal M}_a}[{\cal F}_3[\bar{P}^{(t)}]](x) 
%( \Psi[\Gamma^{(e)}_{{\cal M}_a}[{\cal F}_3[\bar{P}]]](x)-\Psi[\bar{P}^{(t)}](x))\\
%&
+\frac{1}{\gamma}
D_\Psi(P^{0}\|\bar{P}^{(t)})
%\sum_{x \in {\cal X}} P^{0}(x)  
%\Big(\Psi[P^{0}](x)-\Psi[\bar{P}^{(t)}](x)\Big)
-D(\Gamma^{(e)}_{{\cal M}_a}[{\cal F}_3[\bar{P}^{(t')}]]\| \bar{P}^{(t'+1)} ) 
\nonumber \\
\stackrel{(a)}{\ge} &
\frac{1}{\gamma}J_\gamma(\Gamma^{(e)}_{{\cal M}_a}[{\cal F}_3[\bar{P}^{(t')}]],\bar{P}^{(t')})
-\frac{1}{\gamma}{\cal G}(P^{0})-\epsilon_1\stackrel{(b)}{\ge} 0,
\end{align}
where $(a)$ follows from (A2+) and \eqref{XP8}
because the relation $D(P^{0}\| \bar{P}^{(t')}) \le \delta$ follows from the assumption of induction.
$(b)$ follows from the fact that $t=t'$ does not satisfy the condition \eqref{MKX}.
Hence, $ D(P^{0}\| \bar{P}^{(t'+1)}) \le D(P^{0}\| \bar{P}^{(t')}) \le \delta$.

\noindent{\bf Step 4:}\quad
This step aims to show the inequality \eqref{CO8T}
when $t_4 \le t_1-1$, i.e., there exists an integer $t \in [1,t_1-1]$ to satisfy \eqref{MKX}.

Pythagorean theorem guarantees 
\begin{align}
& D( {P}^{(t_4+1)} \| \Gamma^{(e)}_{{\cal M}_a}[{\cal F}_3[\bar{P}^{(t_4)}]] ) 
\nonumber \\
\le &
D( {P}^{(t_4+1)} \| \Gamma^{(e)}_{{\cal M}_a}[{\cal F}_3[\bar{P}^{(t_4)}]] ) 
+ D(  \Gamma^{(e)}_{{\cal M}_a}[{\cal F}_3[\bar{P}^{(t_4)}]]  \| \bar{P}^{(t_4+1)}) \nonumber \\
=& D(P^{(t_4+1)} \| \bar{P}^{(t_4+1)}) \le \epsilon_2
\Label{ZPS}.
\end{align}
Then, we have
\begin{align}
&
{\cal G}( {P}^{(t_2)}) %-\epsilon_3
-\gamma D(P^{(t_2)} \| \bar{P}^{(t_2)})
% \nonumber \\
\stackrel{(a)}{\le} %&
{\cal G}( {P}^{(t_4+1)}) %- \epsilon_3  
-\gamma  D(P^{(t_4+1)} \| \bar{P}^{(t_4+1)})
\nonumber \\
%\stackrel{(a')}{\le} &
%{\cal G}( \bar{P}^{(t_4+1)})   -\gamma  D(P^{(t_4+1)} \| \bar{P}^{(t_4+1)})
%\nonumber \\
\stackrel{(b)}{\le} &
J_\gamma({P}^{(t_4+1)},\bar{P}^{(t_4)}) -\gamma  D(P^{(t_4+1)} \| \bar{P}^{(t_4+1)})
\nonumber \\
\stackrel{(c)}{=} &
J_\gamma(\Gamma^{(e)}_{{\cal M}_a}[{\cal F}_3[\bar{P}^{(t_4)}]],\bar{P}^{(t_4)})
+\gamma D( {P}^{(t_4+1)} \| \Gamma^{(e)}_{{\cal M}_a}[{\cal F}_3[\bar{P}^{(t_4)}]] )%\nonumber \\
-\gamma D(P^{(t_4+1)} \| \bar{P}^{(t_4+1)})
\nonumber \\
\if0
\le &
J_\gamma(\Gamma^{(e)}_{{\cal M}_a}[{\cal F}_3[\bar{P}^{(t_4)}]],\bar{P}^{(t_4)})
+\gamma D( {P}^{(t_4+1)} \| \Gamma^{(e)}_{{\cal M}_a}[{\cal F}_3[\bar{P}^{(t_4)}]] ) \nonumber \\
&+\gamma D(  \Gamma^{(e)}_{{\cal M}_a}[{\cal F}_3[\bar{P}^{(t_4)}]]  \| \bar{P}^{(t_4+1)})
-\gamma D(P^{(t_4+1)} \| \bar{P}^{(t_4+1)})
\nonumber \\
\fi
\stackrel{(d)}{\le} & 
J_\gamma(\Gamma^{(e)}_{{\cal M}_a}[{\cal F}_3[\bar{P}^{(t_4)}]],\bar{P}^{(t_4)}) \Label{XMO},
\end{align}
where each step is derived as follows.
Step $(a)$ follows from the relation
$t_2=$\par\noindent$ \argmin_{t=2, \ldots, t_1} 
D^F(\theta^{(t)} \| \theta_{(t-1)})-D^F(\theta^{(t)} \| \bar\theta^{(t)})$.
%Step $(a')$ follows from \eqref{NXP2}.
Step $(b)$ follows from Lemma \ref{L2} and the condition (A1+)
because 
\eqref{ZPS} holds, and
the relation $D(P^{0}\| \bar{P}^{(t_4)}) \le \delta$ follows from
item (i) with $t=t_4-1$ shown in Step 3.
Step $(c)$ follows from \eqref{XMY2}.
Step $(d)$ follows from the equation \eqref{ZPS}.

Combining \eqref{XMO} and \eqref{MKX}, we have
\begin{align}
{\cal G}( {P}^{(t_2)})
 -\gamma D(P^{(t_2)} \| \bar{P}^{(t_2)})
\le \epsilon_1 %+\epsilon_3 
+{\cal G}(P^{0})\Label{CO8TV},
\end{align}
which implies \eqref{CO8T}.

\noindent{\bf Step 5:}\quad
This step aims to show 
\begin{align}
 J_\gamma(\Gamma^{(e)}_{{\cal M}_a}[{\cal F}_3[\bar{P}^{(t_3)}]],\bar{P}^{(t_3)})
-{\cal G}(P^{0})
-\epsilon_1 
%\nonumber \\
%\le & \frac{\gamma}{t_1-1} (D(P^{0}\| \bar{P}^{(1)})- D(P^{0}\| \bar{P}^{(t_1)}))
\le \frac{\gamma}{t_1-1} D(P^{0}\| \bar{P}^{(1)})
 \Label{CO4}
\end{align}
under the choice of $t_3:= 
\argmin_{1\le t \le t_1-1}
J_\gamma(\Gamma^{(e)}_{{\cal M}_a}[{\cal F}_3[\bar{P}^{(t)}]],\bar{P}^{(t)})$
when $t_4 = t_1$, i.e., there exists no integer $t \in [1,t_1-1]$ to satisfy \eqref{MKX}.

Using \eqref{XMZA},
we have
\begin{align}
\frac{1}{\gamma}J_\gamma(\Gamma^{(e)}_{{\cal M}_a}[{\cal F}_3[\bar{P}^{(t_3)}]],\bar{P}^{(t_3)})
 -\frac{1}{\gamma}{\cal G}(P^{0})
-\epsilon_1 
\le & D(P^{0}\| \bar{P}^{(t)})- D(P^{0}\| \bar{P}^{(t+1)}) \Label{CO1T}
\end{align}
for $t \le t_1-1$.
Taking the sum for \eqref{CO1T}, we have
\begin{align}
&\frac{1}{\gamma}J_\gamma(\Gamma^{(e)}_{{\cal M}_a}[{\cal F}_3[\bar{P}^{(t_3)}]],\bar{P}^{(t_3)})
 -\frac{1}{\gamma}{\cal G}(P^{0})
-\epsilon_1 
\nonumber \\
\le & \frac{1}{t_1-1}\sum_{t=1}^{t=t_1-1}
D(P^{0}\| \bar{P}^{(t)})- D(P^{0}\| \bar{P}^{(t+1)}) \nonumber \\
=& \frac{1}{t_1-1} (
D(P^{0}\| \bar{P}^{(1)})- D(P^{0}\| \bar{P}^{(t_1)}))
\le \frac{1}{t_1-1} D(P^{0}\| \bar{P}^{(1)}).
\end{align}
Therefore, we obtain \eqref{CO4}.

\noindent{\bf Step 6:}\quad
This step aims to show the inequality
\eqref{CO8T}
when $t_4 = t_1$, i.e., there exists no integer $t \in [1,t_1-1]$ to satisfy \eqref{MKX}.
We obtain the following inequality
\begin{align}
{\cal G}( {P}^{(t_2)}) %-\epsilon_3 
-\gamma D(P^{(t_2)} \| \bar{P}^{(t_2)})
\le 
J_\gamma(\Gamma^{(e)}_{{\cal M}_a}[{\cal F}_3[\bar{P}^{(t_3)}]],\bar{P}^{(t_3)}) \Label{XMOW}
\end{align}
in the same way as \eqref{XMO} in Step 4 by changing $t_4$ by $t_3$.
Combining \eqref{XMOW} and \eqref{CO4}, we obtain \eqref{CO8T}.

\end{document}